\PassOptionsToPackage{unicode}{hyperref}
\PassOptionsToPackage{hyphens}{url}
\PassOptionsToPackage{dvipsnames,svgnames,x11names}{xcolor}
\documentclass[
12pt]{article}

\usepackage{graphicx}%
\usepackage{subcaption}
\usepackage{multirow}%
\usepackage{amsmath,amssymb,amsfonts}%
\numberwithin{equation}{section}
\usepackage{amsthm}%
\newtheorem{theorem}{Theorem}[section]

\newtheorem{remark}{Remark}[section]
\newtheorem{corollary}{Corollary}[section]
\newtheorem{proposition}{Proposition}[section]
\newtheorem{definition}{Definition}[section]

\usepackage{enumitem}
\usepackage{mathrsfs}%
\usepackage{textcomp}%
\usepackage{manyfoot}%
\usepackage{booktabs}%
\usepackage{float}
\usepackage{listings}%
\usepackage[ruled,linesnumbered]{algorithm2e}
\usepackage{iftex}
\ifPDFTeX
\usepackage[T1]{fontenc}
\usepackage[utf8]{inputenc}
\usepackage{textcomp} 
\else 
\usepackage{unicode-math}
\defaultfontfeatures{Scale=MatchLowercase}
\defaultfontfeatures[\rmfamily]{Ligatures=TeX,Scale=1}
\fi
\usepackage{lmodern}
\ifPDFTeX\else  
\fi
\IfFileExists{upquote.sty}{\usepackage{upquote}}{}
\IfFileExists{microtype.sty}{
	\usepackage[]{microtype}
	\UseMicrotypeSet[protrusion]{basicmath} 
}{}
\makeatletter
\@ifundefined{KOMAClassName}{
	\IfFileExists{parskip.sty}{%
		\usepackage{parskip}
	}{
		\setlength{\parindent}{0pt}
		\setlength{\parskip}{6pt plus 2pt minus 1pt}}
}{
	\KOMAoptions{parskip=half}}
\makeatother
\usepackage{xcolor}
\makeatletter
\ifx\paragraph\undefined\else
\let\oldparagraph\paragraph
\renewcommand{\paragraph}{
	\@ifstar
	\xxxParagraphStar
	\xxxParagraphNoStar
}
\newcommand{\xxxParagraphStar}[1]{\oldparagraph*{#1}\mbox{}}
\newcommand{\xxxParagraphNoStar}[1]{\oldparagraph{#1}\mbox{}}
\fi
\ifx\subparagraph\undefined\else
\let\oldsubparagraph\subparagraph
\renewcommand{\subparagraph}{
	\@ifstar
	\xxxSubParagraphStar
	\xxxSubParagraphNoStar
}
\newcommand{\xxxSubParagraphStar}[1]{\oldsubparagraph*{#1}\mbox{}}
\newcommand{\xxxSubParagraphNoStar}[1]{\oldsubparagraph{#1}\mbox{}}
\fi
\makeatother

\usepackage{longtable,booktabs,array}
\usepackage{calc} 
\usepackage{etoolbox}
\makeatletter
\patchcmd\longtable{\par}{\if@noskipsec\mbox{}\fi\par}{}{}
\makeatother
\IfFileExists{footnotehyper.sty}{\usepackage{footnotehyper}}{\usepackage{footnote}}
\makesavenoteenv{longtable}
\usepackage{graphicx}
\makeatletter
\def\maxwidth{\ifdim\Gin@nat@width>\linewidth\linewidth\else\Gin@nat@width\fi}
\def\maxheight{\ifdim\Gin@nat@height>\textheight\textheight\else\Gin@nat@height\fi}
\makeatother
\setkeys{Gin}{width=\maxwidth,height=\maxheight,keepaspectratio}
\makeatletter
\def\fps@figure{htbp}
\makeatother

\makeatletter
\@ifpackageloaded{caption}{}{\usepackage{caption}}
\AtBeginDocument{%
	\ifdefined\contentsname
	\renewcommand*\contentsname{Table of contents}
	\else
	\newcommand\contentsname{Table of contents}
	\fi
	\ifdefined\listfigurename
	\renewcommand*\listfigurename{List of Figures}
	\else
	\newcommand\listfigurename{List of Figures}
	\fi
	\ifdefined\listtablename
	\renewcommand*\listtablename{List of Tables}
	\else
	\newcommand\listtablename{List of Tables}
	\fi
	\ifdefined\figurename
	\renewcommand*\figurename{Figure}
	\else
	\newcommand\figurename{Figure}
	\fi
	\ifdefined\tablename
	\renewcommand*\tablename{Table}
	\else
	\newcommand\tablename{Table}
	\fi
}
\@ifpackageloaded{float}{}{\usepackage{float}}
\floatstyle{ruled}
\@ifundefined{c@chapter}{\newfloat{codelisting}{h}{lop}}{\newfloat{codelisting}{h}{lop}[chapter]}
\floatname{codelisting}{Listing}

\makeatother
\makeatletter
\@ifpackageloaded{caption}{}{\usepackage{caption}}
\@ifpackageloaded{subcaption}{}{\usepackage{subcaption}}
\makeatother

\ifLuaTeX
\usepackage{selnolig}  
\fi
\usepackage[]{natbib}
\usepackage{bookmark}

\IfFileExists{xurl.sty}{\usepackage{xurl}}{} 
\hypersetup{
	pdftitle={Title},
	pdfauthor={Author 1; Author 2},
	pdfkeywords={3 to 6 keywords, that do not appear in the title},
	colorlinks=true,
	linkcolor={blue},
	filecolor={Maroon},
	citecolor={Blue},
	urlcolor={Blue},
	pdfcreator={LaTeX via pandoc}}

\newcommand{\anon}{1}

\begin{document}

	\def\spacingset#1{\renewcommand{\baselinestretch}%
		{#1}\small\normalsize} \spacingset{1}


	\if1\anon
	{
		\title{\bf Unifiedly Efficient Inference on All-Dimensional Targets for Large-Scale GLMs}
		\author{Bo Fu\textsuperscript{1}, Dandan Jiang\textsuperscript{1}\thanks{Corresponging author: Dandan Jiang. Email address: jiangdd@mail.xjtu.edu.cn}\\
			\textsuperscript{1} School of Mathematics, Xi'an Jiaotong University}
		\maketitle
	} \fi
	
	\if0\anon
	{
		\bigskip
		\bigskip
		\bigskip
		\begin{center}
			{\LARGE\bf Unifiedly Efficient Inference on All-Dimensional Targets for Large-Scale GLMs}
		\end{center}
		\medskip
	} \fi
	
	\bigskip
	\begin{abstract}
		The scalability of Generalized Linear Models (GLMs) for large-scale, high-dimensional data often forces a trade-off between computational feasibility and statistical accuracy, particularly for inference on pre-specified parameters. While subsampling methods mitigate computational costs, existing estimators are typically constrained by a suboptimal $r^{-1/2}$ convergence rate, where $r$ is the subsample size. This paper introduces a unified framework that systematically breaks this barrier, enabling efficient and precise inference regardless of the dimension of the target parameters. We propose three estimators tailored to different scenarios. For low-dimensional targets, we propose a de-variance subsampling (DVS) estimator that achieves a markedly improved rate of $\max\{r^{-1}, n^{-1/2}\}$,  permitting valid inference even with very small subsamples. As $r$ grows, a multi-step refinement of our estimator is proven to be asymptotically normal and semiparametric efficient when $r/\sqrt{n}\to\infty$, matching the full-sample estimator-confirmed by its Bahadur representation. Critically, for high-dimensional targets, we develop a novel decorrelated score function that facilitates simultaneous inference for a diverging number of pre-specified parameters. Numerical experiments demonstrate that our framework delivers a superior balance of computational efficiency and statistical accuracy across both low- and high-dimensional inferential tasks in large-scale GLM, thereby realizing the promise of unifiedly efficient inference for large-scale GLMs.
	\end{abstract}
	
	\noindent%
	\vfill
	
	\newpage
	\spacingset{1.8} 

\section{Introduction}\label{sec1}
Driven by rapid advancement of technology, data volumes have surged, imposing substantial computational burden and theoretical limitations on conventional approaches. A canonical example arises in regression models involving a scalar response $y$ upon a $p$-dimensional covariate vector $\boldsymbol{x}$. The high dimensionality of the covariate $\boldsymbol{x}$ has motivated the development of high-dimensional statistics; for a comprehensive review, see \citet{buhlmann2011statistics,hastie2015statistical,wainwright2019high}. While a substantial body of literature has focused on estimation and variable selection, the literature on statistical inference has also developed rapidly in recent years. To tackle the high-dimensional inference problem, the debiased Lasso \citep{javanmard2014confidence,van2014asymptotically,zhang2014confidence} and the decorrelated score function \citep{ning2017general,fang2020test} emerge as two key statistical techniques that are both semiparametrically efficient, with numerous extensions to a variety of high-dimensional models \citep{cai2025statistical,cheng2022regularized,yan2023confidence,han2022robust,cai2023statistical}.

In large-scale problems, where both dimensionality $p$ and sample size $n$ grow, computation becomes a major bottleneck. Subsampling provides an effective strategy by selecting a subset of size $r\ll n$ to reduce computational cost. Non-uniform subsampling further improves accuracy by prioritizing informative observations, with designs ranging from leverage-score sampling \citep{ma2014statistical} to methods optimizing asymptotic variance \citep{wang2018optimal,wang2019information,zhang2021optimal,yu2022optimal,wang2021optimal,fan2021optimal,ai2021optimalb}. Although most subsampling estimators converge at rate $r^{-1/2}$, \citet{su2024subsampled} showed that a one-step correction can attain the rate $\max\{n^{-1/2},r^{-1}\}$. However, existing designs mainly target low-dimensional covariates and may fail when $p>r$.

Motivated by the increasing prominence of large-scale data, recent work \citep{shan2024optimal,shao2024optimal} adopts decorrelated-score methods to obtain optimal subsampling estimators in high-dimensional GLMs. Nevertheless, two limitations remain: (i) subsampling sacrifices accuracy relative to full-sample estimators, yet full-sample analysis incurs substantial computational cost; and (ii) existing procedures  target only low-dimensional parameters and do not support high-dimensional simultaneous inference.

In this paper, we address the two aforementioned issues by developing statistically efficient procedures for inference on both low- and high-dimensional parameters in large scale GLMs. Leveraging the one-step idea of \citet{su2024subsampled}, we generalize it to large scale GLMs and propose a de-variance subsampling (DVS) estimator for pre-specified low-dimensional parameters, attaining the convergence rate $\max\{n^{-1/2},r^{-1}\}$ and enabling valid inference from a small subsample. In particular, when $r/\sqrt{n}\to\infty$, the estimator is asymptotically normal and semiparametrically efficient with the convergence rate $n^{-1/2}$. When full-data accuracy is essential, we further introduce a computationally efficient multi-step estimator that is also semiparametrically efficient. Moreover, in contrast to most prior works, which are restricted to inference on low-dimensional parameters, we generalize the decorrelated score function \citep{ning2017general} to high-dimensional parameters of interest by adopting the classic technique \citep{cai2011constrained,yan2023confidence,cai2025statistical} to approximate the inverse Hessian matrix. Since simultaneous inference remains nontrivial in high dimensions, we employ the high-dimensional bootstrap \citep{chernozhukov2023high} to obtain valid confidence regions. Our contributions can be summarized as follows:
\begin{enumerate}
	\item \textbf{A Unifying Inference Framework for Low-Dimensional Targets.} Beyond merely improving upon existing subsampling estimators, we develop a two-stage inference framework that adapts to varying computational budgets and accuracy requirements. Theoretically, we break the long-standing $r^{-1/2}$ barrier through a novel DVS estimator, which achieves a strictly superior   $\max\{n^{-1/2},r^{-1}\}$ rate. This constitutes a fundamental advancement in subsampling theory. Practically, we complement this with a multi-step estimator that transitions seamlessly to full-sample statistical efficiency when more computational resources are available. The two methods are unified under a non-asymptotic Bahadur representation framework, which not only certifies their asymptotic normality and semiparametric efficiency but also provides a coherent principle for selecting the appropriate estimator in practice. This systematic approach resolves the trade-off between computational cost and inferential accuracy that has limited prior subsampling methods.
	\item \textbf{A Scalable Framework for High-Dimensional Simultaneous Inference.} For the challenging task of inference on a diverging number of parameters, we propose a novel subsampled decorrelated score method, which successfully integrates subsampling into high-dimensional simultaneous inference and overcomes the severe computational bottlenecks of full-data methods. By establishing a non-asymptotic Bahadur representation and employing a high-dimensional Gaussian approximation with bootstrap calibration, our method achieves statistical accuracy comparable to the full-sample debiased estimator while offering substantial computational savings. This provides a practical and theoretically sound solution for large-scale simultaneous inference problems previously considered intractable.
\end{enumerate}

The rest of this paper is structured as follows. 
Section \ref{sec2} presents a subsampling framework for large-scale GLMs constructed upon the decorrelated score function. Section \ref{sec3} develops two estimators for the interested low-dimensional parameters  in large-scale GLMs, the de-variance subsampling (DVS) and a multi-step estimator, tailored to different computational budgets for inferring low-dimensional parameters of interest in large-scale GLMs. Section \ref{sec4} extends the decorrelated score approach to enable simultaneous inference on high-dimensional targeted parameters. Simulation studies are reported in Section \ref{sec5}, and a real data analysis is presented in Section \ref{sec6}. Section \ref{sec7} delivers our concluding remarks and some avenues for future exploration. Technical proofs and additional simulation results are relegated to the supplementary materials.\\
\noindent
\textbf{Notation:} For a vector $\mathbf{v}=(v_1,\cdots,v_p)^\top\in\mathbb{R}^{p}$, denote $\|\mathbf{v}\|=(\sum_{i=1}^{p}|v_i|^2)^{1/2}$, $\|\mathbf{v}\|_1=\sum_{i=1}^{p}|v_i|$, $\|\mathbf{v}\|_\infty=\max_{1\leq i\leq p}|v_i|$. For $\mathcal{S}\subseteq\{1,\cdots,p\}$, let $\mathbf{v}_{\mathcal{S}}=\{v_j,j\in\mathcal{S}\}$ and $\nabla_{\mathcal{S}}f(\mathbf{v})$ be the gradient of $f(\mathbf{v})$ with respect to $\mathbf{v}_{\mathcal{S}}$ for a given function $f$. For a matrix $\mathbf{A}=(a_{ij})_{1\leq i\leq k_1,1\leq j\leq k_2}\in \mathbb{R}^{k_1\times k_2}$, denote $\|\mathbf{A}\|_\infty=\max_{1\leq i\leq k_1,1\leq j\leq k_2}|a_{ij}|$, $\|\mathbf{A}\|_{1,1}=\sum_{i=1}^{k_1}\sum_{j=1}^{k_2}|a_{ij}|$, $\|\mathbf{A}\|_{L_1}=\max_{1\leq i\leq k_1}\sum_{j=1}^{k_2}|a_{ij}|$, and let $\|\mathbf{A}\|$ be the spectral norm of $\mathbf{A}$. We use $\otimes$ to depict the Kronecker product. $b'(t)$, $b''(t)$, $b'''(t)$ are the first, second, and third derivatives of $b(t)$, respectively. We define the sub-Gaussian norm and the sub-exponential norm of a random variable $X$ as $\|X\|_{\psi_2}=\inf\{t>0:\mathbb{E}\exp(X^2/t^2)\leq 2\}$ and $\|X\|_{\psi_1}=\inf\{t>0:\mathbb{E}\exp(|X|/t)\leq 2\}$, respectively. For any two sequences $\{a_n\}$ and $\{b_n\}$, we use $a_n=O(b_n),a_n=o(b_n),$ and $a_n\asymp b_n$ to represent that there exists $c,C>0$ such that $|a_n|\leq C|b_n|$ for all $n$, $\lim_{n\to\infty}|a_n|/|b_n|=0$, and $0<c<|a_n/b_n|<C<\infty$, respectively. $a_n\lesssim b_n$ means $a_n=O(b_n)$.  In addition, $O_P$ and $o_p$ have similar meanings as above except that the relationship now holds with high probability.

\section{Preliminaries}\label{sec2}

\subsection{Decorrelated score function for GLMs}\label{sec2.1}
In this section, we present a concise overview of GLMs and the decorrelated score function. Suppose that the observations $\{y_i,\boldsymbol{x}_i\}_{i=1}^n$ are independent and identically distributed (i.i.d.) drawn from the population $\{y,\boldsymbol{x}\}$ with some joint distribution $P$, where $y\in\mathbb{R}$ represents the response variable and \(\boldsymbol{x}\in\mathbb{R}^p\) is the covariate vector. Within the framework of GLMs with a canonical link, the conditional density of $y|\boldsymbol{x}$ can be expressed as:
\begin{align}
	f(y|\boldsymbol{x},\boldsymbol{\beta}_0,\sigma_0)\propto 
	\exp\left\{
	\frac{y\boldsymbol{x}^\top \boldsymbol{\beta}_0 - b(\boldsymbol{x}^\top \boldsymbol{\beta}_0)}{c(\sigma_0)}
	\right\}=\exp\left\{
	\frac{y(\boldsymbol{\theta}_0^\top \boldsymbol{z}+\boldsymbol{\gamma}_0^\top\boldsymbol{u})- b(\boldsymbol{\theta}_0^\top \boldsymbol{z}+\boldsymbol{\gamma}_0^\top\boldsymbol{u})}{c(\sigma_0)}
	\right\},\notag
\end{align}
where $\boldsymbol{\beta}_0=(
\boldsymbol{\theta}_0^\top,
\boldsymbol{\gamma}_0^\top)^\top$ 
is the unknown parameter vector, and $\boldsymbol{x}^\top=(\boldsymbol{z}^\top,\boldsymbol{u}^\top)$. Specifically, the covariates are partitioned into two components: $\boldsymbol{z}\in\mathbb{R}^{d}$ are the covariates corresponding to $\boldsymbol{\theta}$ that is of interest and $d\leq p$, and $\boldsymbol{u}\in\mathbb{R}^{p-d}$ are the covariates corresponding to nuisance parameter $\boldsymbol{\gamma}$. Though $d$ is assumed fixed in many previous works \citep{ning2017general,shan2024optimal,shao2024optimal}, we do not impose this condition in this work. The functions $b(\cdot)$ and $c(\cdot)$ are known functions, with $\sigma_0$ a known dispersion parameter $\sigma$, which is assumed to be fixed. In this formulation, we concentrate exclusively on $\boldsymbol{\theta}$. To streamline the notation, the intercept term is subsumed into $\boldsymbol{\gamma}_0$ without loss of generality.

Consider the negative log-likelihood function $l(\boldsymbol{\beta}_0)=l(\boldsymbol{\theta}_0,\boldsymbol{\gamma}_0)\propto -\log f(y|\boldsymbol{x},\boldsymbol{\beta}_0,\sigma_0)=b(\boldsymbol{\theta}_0^\top \boldsymbol{z}+\boldsymbol{\gamma}_0^\top\boldsymbol{u})-y(\boldsymbol{\theta}_0^\top \boldsymbol{z}+\boldsymbol{\gamma}_0^\top\boldsymbol{u}).$ The decorrelated score function in \citet{ning2017general} takes the form $S(\boldsymbol{\beta},\mathbf{w}_0)=\nabla_\theta l(\boldsymbol{\boldsymbol{\beta}})-\mathbf{w}_0\nabla_\gamma l(\boldsymbol{\boldsymbol{\beta}})$, where 
\begin{align*}
	\mathbf{w}_0&=\mathop{\arg\min}\limits_{\mathbf{w}\in\mathbb{R}^{d\times (p-d)}}\mathbb{E}[\|\nabla_{\boldsymbol{\theta}}l(\boldsymbol{\beta}_{0})-\mathbf{w}\nabla_{\boldsymbol{\gamma}}l(\boldsymbol{\beta}_{0})\|]^2=\mathop{\arg\min}\limits_{\mathbf{w}\in\mathbb{R}^{d\times (p-d)}}\mathbb{E}[b''(\boldsymbol{\beta}_0^\top\boldsymbol{x})\|\boldsymbol{z}-\mathbf{w}\boldsymbol{u}\|^2]\\&=\mathbb{E}[b''(\boldsymbol{\beta}_0^\top\boldsymbol{x})\boldsymbol{z}\boldsymbol{u}^\top]\left(\mathbb{E}[b''(\boldsymbol{\beta}_0^\top\boldsymbol{x})\boldsymbol{u}\boldsymbol{u}^\top]\right)^{-1}.
\end{align*}

The decorrelated score function offers two merits. First, $\boldsymbol{\gamma}$ does not influence the decorrelated score function around $\boldsymbol{\theta}_{0}$ since $\mathbb{E}[\nabla_{\boldsymbol{\gamma}}\boldsymbol S(\boldsymbol{\beta}_{0},\mathbf{w}_0)]=\mathbf{0}_{d\times (p-d)}$, ensuring the convergence rate of $\boldsymbol{\theta}$ is not affected by the high-dimensional nuisance parameters $\boldsymbol{\gamma}$. Second, it is orthogonal to the nuisance score function, as expressed by $\mathbb{E}[S(\boldsymbol{\beta}_{0},\mathbf{w}_0)^\top\nabla_{\boldsymbol{\gamma}}l(\boldsymbol{\beta}_{0})]=\mathbf{0}_{d\times (p-d)}$.
For the i.i.d. observations $\{y_i,\boldsymbol{x}_i\}_{i=1}^n$, the full-data-based decorrelated score function is $$S(\boldsymbol{\theta},\boldsymbol{\gamma}_0,\mathbf{w}_0)=\frac{1}{n}\sum_{i=1}^{n}S_i(\boldsymbol{\theta},\boldsymbol{\gamma}_0,\mathbf{w}_0)=\frac{1}{n}\sum_{i=1}^{n}\left(b'(\boldsymbol{\theta}^\top \boldsymbol{z}_{i}+\boldsymbol{\gamma}_0^\top\boldsymbol{u}_{i})-y_i\right)\left(\boldsymbol{z}_{i}-\mathbf{w}_0\boldsymbol{u}_{i}\right).$$

\subsection{Subsampling-based decorrelated score function}
Now we introduce the subsampling decorrelated score function implemented via Poisson subsampling. Denote $\mathcal{K}$ as the index set of a subsample attained via Poisson subsampling defined later, the corresponding subsampling-based decorrelated score function is
\begin{align*}
	S^*(\boldsymbol{\theta},\boldsymbol{\gamma}_0,\mathbf{w}_0)&=\frac{1}{r}\sum_{i\in\mathcal{K}}S_i(\boldsymbol{\theta},\boldsymbol{\gamma}_0,\mathbf{w}_0)=\frac{1}{r}\sum_{i\in\mathcal{K}}\left(b'(\boldsymbol{\theta}^\top \boldsymbol{z}_{i}+\boldsymbol{\gamma}_0^\top\boldsymbol{u}_{i})-y_i\right)\left(\boldsymbol{z}_{i}-\mathbf{w}_0\boldsymbol{u}_{i}\right).
\end{align*}

Compared to full-data-based approaches, subsampling offers a practical solution to alleviate the computational burden. Here we adopt Poisson subsampling \citep{yu2022optimal,yao2023optimal,shan2024optimal}, a technique that has garnered significant interest due to its computational efficiency relative to sampling with replacement. Specifically, let $\mathcal{P}$ be the index set of subsampled data points, then uniform Poisson subsampling involves generating $n$ i.i.d. Bernoulli random variables $\delta_i,i=1,\cdots,n$ to determine whether the $i$-th data point is selected, i.e., $\mathbb{P}(\delta_i=1)=r/n$, where $r$ is the expected subsample size. The general procedure for uniform Poisson subsampling is presented in Algorithm \ref{alg1}. 
\begin{algorithm}[htbp]
	\SetAlgoLined
	{\bf Initialization} $\mathcal{P}=\varnothing$, $p=r/n$, where $r$ is the expected subsample size\;
	\For{$i=1,\cdots,n$}{
		Generate a Bernoulli variable $\delta_i\sim\text{Bernoulli}(p)$\;
		\textbf{If} $\delta_i=1$ \textbf{do:} Update $\mathcal{P}=\mathcal{P}\cup\{(y_{i},\boldsymbol{x}_i)\}$\;
	}
	{\bf Output} subsample $\mathcal{P}$.
	\caption{General uniform Poisson subsampling algorithm}\label{alg1}
\end{algorithm}
A key merit of Poisson subsampling is that the inclusion decision for each data $(y_i,\boldsymbol{x}_i)$ hinges only on $\delta_i$, with the $\delta_i$'s, for $i=1,\cdots,n$, being mutually independent. Moreover, $\delta_i$ can be generated either sequentially or in blocks, allowing for efficient parallel processing. We assume $r\leq n$ throughout this paper, a condition that naturally applies in the setting of large-scale datasets. When $r=n$, the subsample degenerates to the original dataset.

Since both the full-data-based decorrelated score function $S(\boldsymbol{\theta},\boldsymbol{\gamma}_0,\mathbf{w}_0)$ and the subsampled version $S^*(\boldsymbol{\theta},\boldsymbol{\gamma}_0,\mathbf{w}_0)$ depend on the unknown parameters $\boldsymbol{\gamma}_0$ and $\mathbf{w}_0$, we draw another subsample $\mathcal{K}_\mathrm{p}$ with an expected subsample size of $r_\mathrm{p}$ through Algorithm \ref{alg1} to obtain their initial estimators similar to \citet{shan2024optimal}. Specifically, the two Lasso type estimators $\hat{\boldsymbol{\beta}}_\mathrm{p},\hat{\mathbf{w}}_\mathrm{p}$ can be calculated by employing \texttt{R} packages \texttt{glmnet} and \texttt{pqr}:
\begin{align}
	&\hat{\boldsymbol{\beta}}_\mathrm{p}=(\hat{\boldsymbol{\theta}}_\mathrm{p}^\top,\hat{\boldsymbol{\gamma}}_\mathrm{p}^\top)^\top=\mathop{\arg\min}\limits_{\boldsymbol{\beta}\in\mathbb{R}^{p}}\left\{\frac{1}{r_\mathrm{p}}\sum_{i\in\mathcal{K}_\mathrm{p}}\left(b(\boldsymbol{\beta}^\top\boldsymbol{x}_i^{*\mathrm{p}})-y_i^{*\mathrm{p}}\boldsymbol{\beta}^\top\boldsymbol{x}_i^{*\mathrm{p}}\right)+\lambda\|\boldsymbol{\beta}\|_1\right\},\label{eq1}\\
	&\hat{\mathbf{w}}_\mathrm{p}=\mathop{\arg\min}\limits_{\mathbf{w}\in\mathbb{R}^{d\times(p-d)}}\left\{\frac{1}{r_\mathrm{p}}\sum_{i\in\mathcal{K}_\mathrm{p}}b''(\hat{\boldsymbol{\beta}}_\mathrm{p}^\top\boldsymbol{x}_i^{*\mathrm{p}})\|\boldsymbol{z}_i^{*\mathrm{p}}-\mathbf{w}\boldsymbol{u}_i^{*\mathrm{p}}\|^2+\tau\sum_{j=1}^{p-d}\|\mathbf{w}_j\|_1\right\},\label{eq2}
\end{align}
where $\lambda,\tau$ are two regularized parameters and $\mathbf{w}_j$ is the $j$-th column of $\mathbf{w}$. Then the full-data-based estimator $\hat{\boldsymbol{\theta}}_{full}\in\mathbb{R}^{d}$ is the solution to the equation below:
\begin{align}\label{eq3}
	S(\boldsymbol{\theta},\hat{\boldsymbol{\gamma}}_\mathrm{p},\hat{\mathbf{w}}_\mathrm{p})=\frac{1}{n}\sum_{i=1}^{n}S_i(\boldsymbol{\theta},\hat{\boldsymbol{\gamma}}_\mathrm{p},\hat{\mathbf{w}}_\mathrm{p})=\frac{1}{n}\sum_{i=1}^{n}\left(b'(\boldsymbol{\theta}^\top \boldsymbol{z}_{i}+\hat{\boldsymbol{\gamma}}_\mathrm{p}^\top\boldsymbol{u}_{i})-y_i\right)\left(\boldsymbol{z}_i-\hat{\mathbf{w}}_\mathrm{p}\boldsymbol{u}_i\right)=\mathbf{0},
\end{align}
and the subsampling-based estimator 
$\hat{\boldsymbol{\theta}}_{uni}\in\mathbb{R}^{d}$ is the solution to the equation below:
\begin{equation}
	\begin{aligned}\label{eq4}
		\mathbf{0}=S^*(\boldsymbol{\theta},\hat{\boldsymbol{\gamma}}_\mathrm{p},\hat{\mathbf{w}}_\mathrm{p})&=\frac{1}{r}\sum_{i\in\mathcal{K}}S_i(\boldsymbol{\theta},\hat{\boldsymbol{\gamma}}_\mathrm{p},\hat{\mathbf{w}}_\mathrm{p})=\frac{1}{r}\sum_{i\in\mathcal{K}}\left(b'(\boldsymbol{\theta}^\top \boldsymbol{z}_{i}+\hat{\boldsymbol{\gamma}}_\mathrm{p}^\top\boldsymbol{u}_{i})-y_i\right)\left(\boldsymbol{z}_i-\hat{\mathbf{w}}_\mathrm{p}\boldsymbol{u}_i\right)\\
		&=\frac{1}{n}\sum_{i=1}^{n}\frac{\delta_i}{r/n}\left(b'(\hat{\boldsymbol{\theta}}_\mathrm{p}^\top \boldsymbol{z}_{i}+\hat{\boldsymbol{\gamma}}_\mathrm{p}^\top\boldsymbol{u}_{i})-y_i\right)\left(\boldsymbol{z}_{i}-\mathbf{w}_0\boldsymbol{u}_{i}\right),
	\end{aligned}
\end{equation}
where $\mathcal{K}=\{i|\delta_i=1,i\in\{1,\cdots,n\}\}$ as the index set of subsample attained by performing Algorithm \ref{alg1} with expected subsample size $r$.

However, the subsampled estimator $\hat{\boldsymbol{\theta}}_{uni}$ defined in \eqref{eq4} only achieves $r^{-1/2}$ rate \citep{shan2024optimal}, implying a loss of accuracy compared to the full-data estimator. To perform efficient inference without compromising statistical accuracy, we explore methods to perform fast inference via subsampling techniques in the following sections.

\section{Inference for low-dimensional parameters of interest}\label{sec3}
In this section, we introduce two estimators for efficient inference of low-dimensional parameters of interest in large scale generalized linear models. Both estimators leverage suitably chosen subsample sizes to retain full-data statistical accuracy, while substantially reducing computational cost. The first, referred to as the DVS estimator, emphasizes computational speed and flexibility, making it well-suited for scenarios where small pilot and main subsample sizes ($r_\mathrm{p}$ and $r$) are used. The second, a multi-step estimator, is more computationally economical when full-sample accuracy is desired.
\subsection{Inference via the DVS estimator}\label{sec3.1}
As an intuitive statement, \citet{ning2017general} considered a one-step estimator based on an initial Lasso estimator $\hat{\boldsymbol{\theta}}_{\mathrm{p}}$, taking the following form:
\begin{align}\label{eq5}
	\hat{\boldsymbol{\theta}}_{\mathrm{p}}-\left(\nabla_{\boldsymbol{\theta}}S(\hat{\boldsymbol{\theta}}_{\mathrm{p}},\hat{\boldsymbol{\gamma}}_\mathrm{p},\hat{\mathbf{w}}_\mathrm{p})\right)^{-1}S(\hat{\boldsymbol{\theta}}_{\mathrm{p}},\hat{\boldsymbol{\gamma}}_\mathrm{p},\hat{\mathbf{w}}_\mathrm{p}),
\end{align}
except that $\hat{\boldsymbol{\theta}}_{\mathrm{p}}$ in their work is the Lasso estimator based on the full sample. In contrast, we consider the one-step estimation based on $\hat{\boldsymbol{\theta}}_{uni}$ attained via \eqref{eq4}, which is asymptotically normal \citep{shan2024optimal}. The estimator is formally defined as:

\begin{definition}[de-variance subsampling (\textit{DVS}) estimator]\label{def1}
	Continuing the notations in Section \ref{sec2}, define the DVS estimator as
	\begin{equation}
		\begin{aligned}\label{eq6}
			\tilde{\boldsymbol{\theta}}=&\hat{\boldsymbol{\theta}}_{uni}-\left(\nabla_{\boldsymbol{\theta}}S^*(\hat{\boldsymbol{\theta}}_{uni},\hat{\boldsymbol{\gamma}}_\mathrm{p},\hat{\mathbf{w}}_\mathrm{p})\right)^{-1}S(\hat{\boldsymbol{\theta}}_{uni},\hat{\boldsymbol{\gamma}}_\mathrm{p},\hat{\mathbf{w}}_\mathrm{p}),
		\end{aligned}
	\end{equation}
	where $\nabla_{\boldsymbol{\theta}}S^*(\hat{\boldsymbol{\theta}}_{uni},\hat{\boldsymbol{\gamma}}_\mathrm{p},\hat{\mathbf{w}}_\mathrm{p})=(1/r)\sum_{i\in\mathcal{K}}b''(\hat{\boldsymbol{\theta}}_{uni}^\top \boldsymbol{z}_{i}+\hat{\boldsymbol{\gamma}}_\mathrm{p}^\top\boldsymbol{u}_{i})\left(\boldsymbol{z}_i-\hat{\mathbf{w}}_\mathrm{p}\boldsymbol{u}_i\right){\boldsymbol{z}_{i}}^\top$, and $S(\hat{\boldsymbol{\theta}}_{uni},\hat{\boldsymbol{\gamma}}_\mathrm{p},\hat{\mathbf{w}}_\mathrm{p})=(1/n)\sum_{i=1}^{n}\left(b'(\hat{\boldsymbol{\theta}}_{uni}^\top \boldsymbol{z}_{i}+\hat{\boldsymbol{\gamma}}_\mathrm{p}^\top\boldsymbol{u}_{i})-y_i\right)\left(\boldsymbol{z}_i-\hat{\mathbf{w}}_\mathrm{p}\boldsymbol{u}_i\right)$.
\end{definition}
Note that the additional cost of \eqref{eq6} is $O(n(p-d)d)$, which remains small since even computing the L-optimal subsampling probabilities in \citet{shao2024optimal,shan2024optimal} already requires $O(n(p-d)d)$. Moreover, $\boldsymbol{\Phi}:=\mathbb{E}[b''(\boldsymbol{\beta}_0^\top\boldsymbol{x})(\boldsymbol{z}-\mathbf{w}_0\boldsymbol{u})\boldsymbol{z}^\top]=\mathbb{E}[b''(\boldsymbol{\beta}_0^\top\boldsymbol{x})(\boldsymbol{z}-\mathbf{w}_0\boldsymbol{u})(\boldsymbol{z}-\mathbf{w}_0\boldsymbol{u})^\top]$ coincides with the submatrix of $(\mathbb{E}[b''(\boldsymbol{\beta}_0^\top\boldsymbol{x})\boldsymbol{x}\boldsymbol{x}^\top])^{-1}$ corresponding to the coordinates of $\boldsymbol{z}$, and is therefore symmetric. Thus, the DVS estimator resembles standard debiased estimators for low-dimensional targets, where the main task is approximating the inverse Hessian, differing mainly in that it builds on $\hat{\boldsymbol{\theta}}_{uni}$ rather than $\hat{\boldsymbol{\theta}}_{\mathrm{p}}$ and incorporates decorrelated score functions with subsampling for computational gains. As shown in \citet{ning2017general}, the full-data decorrelated-score estimator is semiparametrically efficient, in line with debiased estimators in \citet{van2014asymptotically,zhang2014confidence,javanmard2014confidence}. Building on this insight, we use the asymptotic distribution of $\hat{\boldsymbol{\theta}}_{uni}$ to perform inference with reduced pilot and main subsample sizes ($r_\mathrm{p}$ and $r$), yielding substantial computational savings while achieving semiparametric efficiency when $r/\sqrt{n}\to\infty$.

To derive the statistical properties of the proposed DVS estimator, we impose the following regularity conditions, which are commonly assumed in high-dimensional GLMs.\\
\noindent
(A.1) For some constants $C_1,C_2>0$, $\lambda_{\min}(\mathbb{E}[b''(\boldsymbol{\beta}_0^\top\boldsymbol{x})\boldsymbol{x}\boldsymbol{x}^\top])\geq C_1$, $\lambda_{\max}(\mathbb{E}[b''(\boldsymbol{\beta}_0^\top\boldsymbol{x})\boldsymbol{x}\boldsymbol{x}^\top])\leq C_2$, where $\lambda_{\min}(\cdot)$ and $\lambda_{\max}(\cdot)$ are the smallest eigenvalues of a given matrix.\\
(A.2) $\boldsymbol{\beta}_0$ is sparse with support $\mathcal{S}_{\boldsymbol{\beta}_0}$ and $|\mathcal{S}_{\boldsymbol{\beta}_0}|=s_1$. The matrix $\mathbf{w}_0$ is sparse with support $\mathcal{S}_{\mathbf{w}_0}=\{j:\mathbf{w}_{0,j}\neq\mathbf{0},j\in[p-d]\}$ and $|\mathcal{S}_{\mathbf{w}_0}|=s_2$.\\
(A.3) $\|\boldsymbol{x}\|_\infty\leq 
C$, $\|y-b'(\boldsymbol{\beta}_0^\top\boldsymbol{x})\|_{\psi_1}\leq C$, and $\max_{1\leq k\leq d}|(\mathbf{w}_0)_{k}\boldsymbol{u}|\leq C$ for some constant $C$, where $(\mathbf{w}_0)_{k}$ is the $k$-th row of $\mathbf{w}_0$.\\
(A.4) There exists some constants $R_1\leq R_2$ such that $\boldsymbol{\beta}_0^\top\boldsymbol{x}\in [R_1,R_2]$. $b$ is third times differentiable and for any $t\in[R_1-\varepsilon,R_2+\varepsilon]$ with some constant $\varepsilon>0$ and a sequence $t_1$ satisfying $|t_1-t|=o(1)$, it holds that $0<b''(t)\leq C$, $|b''(t_1)-b''(t)|\leq C|t_1-t|b''(t)$, $|b'''(t)|\leq C$, and $|b'''(t_1)-b'''(t)|\leq C|t_1-t|$ for some constant $C$.\\
(A.5) $\boldsymbol{\Phi}=\mathbb{E}[b''(\boldsymbol{\beta}_0^\top\boldsymbol{x})(\boldsymbol{z}-\mathbf{w}_0\boldsymbol{u})\boldsymbol{z}^\top]$ is finite and $\lambda_{\min}(\boldsymbol{\Phi})>c$ for some positive constants $c$.\\
(A.6) $r^{-1}\log p=o(1)$, $r_\mathrm{p}^{-1/2}(s_1\vee s_2)\log p=o(1)$, $r^{1/2}r_\mathrm{p}^{-1}(s_1\vee s_2)\log p=o(1)$, $(s_1\vee s_2)\sqrt{\log p/r}=o(1)$. The regularized parameters $\lambda$ and $\tau$ are $\lambda\asymp\tau\asymp \sqrt{\log p/r_\mathrm{p}}$.

Assumptions (A.1)-(A.3) and Assumption (A.6) are also assumed by \citet{shao2024optimal}, in which Assumptions (A.1)-(A.3) are common regular conditions for high-dimensional GLMs \citep{ning2017general}. Assumption (A.1) imposes an eigenvalue constraint on the design matrix, necessitating that the eigenvalues of $\mathbb{E}[\boldsymbol{x}\boldsymbol{x}^\top]$ are strictly bounded between $0$ and $\infty$, a requirement satisfied when $0<C_1<b''(\boldsymbol{\beta}_0^\top\boldsymbol{x})<C_2$ for some constants $C_1$ and $C_2$, and this is encompassed by Assumption (A.4). Assumption (A.2) imposes sparsity conditions on both $\boldsymbol{\beta}_0$ and $\mathbf{w}_0$. 
For the sake of technical simplicity, Assumption (A.3) supposes that the design matrix is bounded and the residuals exhibit sub-exponential tails, which are similarly assumed by \citet{ning2017general,shao2024optimal}. Assumption (A.4) imposes a regularity condition on $b(t)$, which holds for a broad class of GLMs. Assumption (A.5) is similar to the condition in \cite[Theorem 1]{shao2024optimal} and can be directly deduced from Assumption (A.1) together with the Cauchy interlacing theorem in fact. Finally, Assumption (A.6) imposes a minimal restriction on the subsampling size and is non-restrictive.

\begin{remark}\label{remark3.1}
	The sparsity condition for $\mathbf{w}_0$ can be changed into $\max_{1\leq k\leq d}\|(\mathbf{w}_0)_k\|_0=s_2$ if one applies the standard Lasso estimator or the Dantzig type estimator to obtain the estimator of $\mathbf{w}_0$ \citep{cheng2022regularized,ning2017general}. From this sight of view, \citet{ning2017general} states that this condition is identical to the degree of the node $\boldsymbol{Z}$ in the graph if $\boldsymbol{x}$ follows a Gaussian graphical model. This is similar to the assumptions in \citet{van2014asymptotically} if we only want to debias the parameters of interest.
\end{remark}

The following theorem states the properties of the DVS estimator, indicating that it can achieve a faster convergence rate than previous subsampling-based estimators.

\begin{theorem}\label{thm1}
	Under Assumptions (A.1)-(A.6), we have
	\begin{equation}
		\begin{aligned}\label{eq7}
			\tilde{\boldsymbol{\theta}}-\boldsymbol{\theta}_0=O_P\Big(n^{-1/2}+r^{-1}+\sqrt{s_1s_2}r_\mathrm{p}^{-1}\log p+r^{-1/2}r_\mathrm{p}^{-1}s_1^2\log p)\Big).
		\end{aligned}
	\end{equation}
	Furthermore, if $r_\mathrm{p}$ satisfies
	\begin{equation}
		\begin{aligned}\label{eq8}
			\min\{\sqrt{n},r\}\max\left\{\sqrt{s_1s_2}r_\mathrm{p}^{-1}\log p,r^{-1/2}r_\mathrm{p}^{-1}s_1^2\log p,\sqrt{r^{-1}s_2r_\mathrm{p}^{-1}\log p}\right\}=o(1),
		\end{aligned}
	\end{equation}
	the following holds: 
	\begin{equation}
		\begin{aligned}\label{eq9}
			\min\{\sqrt{n},r\}(\tilde{\boldsymbol{\theta}}-\boldsymbol{\theta}_0)\mathop{\to}\limits^{d}h(\boldsymbol{U}),
		\end{aligned}
	\end{equation}
	where $\boldsymbol{U}$ is a $(d^2+2d)$-dimensional random vector with distribution $\mathbb{N}(\mathbf{0}_{d^2+2d},\boldsymbol{V})$, $\boldsymbol{V}$ is partitioned as $[\boldsymbol{V}_{i,j}]_{i,j=1,2,3}$ with
	$$\left\{
	\begin{aligned}
		\boldsymbol{V}_{11}&=\boldsymbol{V}_{22}=c(\sigma_0)\mathbb{E}[b''(\boldsymbol{\beta}_0^\top \boldsymbol{x})(\boldsymbol{z}-\mathbf{w}_0\boldsymbol{u})(\boldsymbol{z}-\mathbf{w}_0\boldsymbol{u})^\top]=c(\sigma_0)\boldsymbol{\Phi},\\
		\boldsymbol{V}_{12}&=\boldsymbol{V}_{21}^\top=\sqrt{r/n}\boldsymbol{V}_{11},\,\boldsymbol{V}_{13}=\boldsymbol{V}_{23}=\boldsymbol{V}_{31}^\top=\boldsymbol{V}_{32}^\top=\mathbf{0}_{d\times d^2},\\
		\boldsymbol{V}_{33}&=(1-\frac{r}{n})\mathbb{E}[(b''(\boldsymbol{\beta}_0^\top\boldsymbol{x}))^2\mbox{vec}\{(\boldsymbol{z}-\mathbf{w}_0\boldsymbol{u})\boldsymbol{z}^\top\}\mbox{vec}\{(\boldsymbol{z}-\mathbf{w}_0\boldsymbol{u})\boldsymbol{z}^\top\}^\top],
	\end{aligned}
	\right.$$
	and
	\begin{align}\label{eq10}
		h(\boldsymbol{U})=\frac{1}{2}m_2\boldsymbol{\Phi}^{-1}\mathcal{T}(\boldsymbol{\theta}_0,\boldsymbol{\gamma}_0,\mathbf{w}_0)((\boldsymbol{\Phi}^{-1}\boldsymbol{U}_1)\otimes(\boldsymbol{\Phi}^{-1}\boldsymbol{U}_1))-m_1\boldsymbol{\Phi}^{-1}\boldsymbol{U}_2-m_2\boldsymbol{\Phi}^{-1}\boldsymbol{U}_3\boldsymbol{\Phi}^{-1}\boldsymbol{U}_1,
	\end{align}
	where $m_1=\min\{\sqrt{n},r\}/\sqrt{n},m_2=\min\{\sqrt{n},r\}/r$, $\mathcal{T}(\boldsymbol{\theta}_0,\boldsymbol{\gamma}_0,\mathbf{w}_0)$ is a $d\times d^2$ matrix with its $j$-th row being $\mbox{vec}\{\mathbb{E}[b'''(\boldsymbol{\beta}_0^\top\boldsymbol{x})(z_{j}-(\mathbf{w}_0)_{j}\boldsymbol{u})\boldsymbol{z}\boldsymbol{z}^\top]\}$, and $\boldsymbol{U}_1\in\mathbb{R}^{d},\boldsymbol{U}_2\in\mathbb{R}^{d},\boldsymbol{U}_3\in\mathbb{R}^{d^2}$ denote the first $d$, next $d$ and last $d^2$ components of $\boldsymbol{U}$, respectively.\\
\end{theorem}
From Theorem \ref{thm1}, we have $\tilde{\boldsymbol{\theta}}-\boldsymbol{\theta}_0=O_P(n^{-1/2}+r^{-1}+\sqrt{s_1s_2}r_\mathrm{p}^{-1}\log p)$ if $r^{-1/2}s_1^{3/2}s_2^{-1/2}=o(1)$, which demonstrates that $\tilde{\boldsymbol{\theta}}-\boldsymbol{\theta}_0$ is composed of three components. The first component is of order $n^{-1/2}$, which is the convergence rate of the full-data-based estimator. The second component is of order $r^{-1}$, which is faster than the convergence rates of the subsampled estimators in previous work that are typically $r^{-1/2}$. The third component is $O_P(\sqrt{s_1s_2}r_\mathrm{p}^{-1}\log p)$, which is the error brought by the initial estimators $\hat{\boldsymbol{\beta}}_\mathrm{p}$ and $\hat{\mathbf{w}}_\mathrm{p}$. However, even when we choose $r_\mathrm{p}=r$, the convergence rate of $\tilde{\boldsymbol{\theta}}$ is $\max\{n^{-1/2},\sqrt{s_1s_2}r^{-1}\log p\}$, showing that the convergence rate is always faster than $r^{-1/2}$ when $\sqrt{s_1s_2/r}\log p=o(1)$ and the error is inversely proportional to $r$ when $r\ll\sqrt{n}$, rather than being proportional to $r^{-1/2}$.

For the asymptotic distribution of the DVS estimator $\tilde{\boldsymbol{\theta}}$, we need a larger $r_\mathrm{p}$ to ensure the accuracy. To guide practical implementations for constructing confidence intervals, we categorize the relationship between $r$ and $n$ into three scenarios to determine how to select pilot subsample size $r_\mathrm{p}$. This division facilitates flexible choices of the subsample sizes $r$ and $r_\mathrm{p}$ to strike a balance between computational cost and the length of confidence interval. We first consider the case when $r/\sqrt{n}\to\infty$. The following theorem exhibits the non-asymptotic Bahadur representation and the asymptotic distribution of the DVS estimator $\tilde{\boldsymbol{\theta}}$. It verifies the semiparametric efficiency \citep{van2000asymptotic} of the DVS estimator under $r/\sqrt{n}\to\infty$.

\begin{theorem}[Case of $r/\sqrt{n}\to\infty$]\label{thm2}
	Under Assumptions (A.1)-(A.6), we have 
	\begin{align*}
		&\sqrt{n}\{(\tilde{\boldsymbol{\theta}}-\boldsymbol{\theta}_0)+\boldsymbol{\Phi}^{-1}S(\boldsymbol{\theta}_0,\boldsymbol{\gamma}_0,\mathbf{w}_0)\}\\
		=&O_P(n^{1/2}r^{-1}+n^{1/2}(\sqrt{s_1s_2}\vee s_2)\log p/r_\mathrm{p}+n^{1/2}r^{-1/2}r_\mathrm{p}^{-1}s_1^2\log p)+o_P(1).
	\end{align*}
	
	If $r/\sqrt{n}\to\infty$ and $r_\mathrm{p}$ satisfies 
	\begin{equation}\label{eq11}
		\begin{aligned}
			r_\mathrm{p}/(\max\{n^{1/2}(\sqrt{s_1s_2}\vee s_2)\log p,n^{1/2}r^{-1/2}s_1^2\log p\}\to\infty,
		\end{aligned}
	\end{equation}
	we have $
		\sqrt{n}(\tilde{\boldsymbol{\theta}}-\boldsymbol{\theta}_0)\mathop{\to}\limits^{d} N(0,c(\sigma_0)\boldsymbol{\Phi}^{-1})$, which is the asymptotic distribution of the solution to $S(\boldsymbol{\theta},\boldsymbol{\gamma}_0,\mathbf{w}_0)=\mathbf{0}$.
\end{theorem}

While Theorem \ref{thm2} covers the case $r/\sqrt{n}\to\infty$, the other two scenarios are shown in Corollaries \ref{cor1}-\ref{cor2} below, using the notation of Theorem \ref{thm1}.

\begin{corollary}[Case of $r/\sqrt{n}\to 0$]\label{cor1}
	Under Assumptions (A.1)-(A.6), if $r/\sqrt{n}\to0$ and $r_\mathrm{p}$ satisfies $r_\mathrm{p}/(\max\{r\sqrt{s_1s_2}\log p,r^{1/2}s_1^2\log p,rs_2\log p\}\to\infty,$
	we have 
	$$r(\tilde{\boldsymbol{\theta}}-\boldsymbol{\theta}_0)\mathop{\to}\limits^{d}\frac{1}{2}\boldsymbol{\Phi}^{-1}\mathcal{T}(\boldsymbol{\theta}_0,\boldsymbol{\gamma}_0,\mathbf{w}_0)((\boldsymbol{\Phi}^{-1}\boldsymbol{U}_1)\otimes(\boldsymbol{\Phi}^{-1}\boldsymbol{U}_1))-\boldsymbol{\Phi}^{-1}\boldsymbol{U}_3\boldsymbol{\Phi}^{-1}\boldsymbol{U}_1.$$
\end{corollary}

\begin{corollary}[Case of $r/\sqrt{n}\to C\in(0,\infty)$]\label{cor2}
	Under Assumptions (A.1)-(A.6), if $r/\sqrt{n}\to C\in(0,\infty)$ and $r_\mathrm{p}$ satisfies \eqref{eq11}, we have 
	$$\sqrt{n}(\tilde{\boldsymbol{\theta}}-\boldsymbol{\theta}_0)\mathop{\to}\limits^{d}\frac{\sqrt{n}}{2r}\boldsymbol{\Phi}^{-1}\mathcal{T}(\boldsymbol{\theta}_0,\boldsymbol{\gamma}_0,\mathbf{w}_0)((\boldsymbol{\Phi}^{-1}\boldsymbol{U}_1)\otimes(\boldsymbol{\Phi}^{-1}\boldsymbol{U}_1))-\boldsymbol{\Phi}^{-1}\boldsymbol{U}_2-\frac{\sqrt{n}}{r}\boldsymbol{\Phi}^{-1}\boldsymbol{U}_3\boldsymbol{\Phi}^{-1}\boldsymbol{U}_1.$$
\end{corollary}

Corollary \ref{cor1} establishes that for $r/\sqrt{n}\to 0$, the DVS estimator $\tilde{\boldsymbol{\theta}}$ achieves a convergence rate of $r$ and presents the detailed asymptotic distribution. Additionally, Corollary \ref{cor2} further indicates that when $r/\sqrt{n}\to C\in(0,\infty)$, the asymptotic distribution takes the form of a mixture between a Gaussian component (as in the case $r/\sqrt{n}\to\infty$) and an additional term linked to the dominant item under the case of $r/\sqrt{n}\to 0$. 
\begin{algorithm}[htbp]
	\SetAlgoLined
	{\bf Pilot Subsampling:} Run Algorithm \ref{alg1} with expected subsample size $r_\mathrm{p}$ to get a initial subsample set $\mathcal{K}_\mathrm{p}$, then estimate  $\hat{\boldsymbol{\beta}}_{\mathrm{p}}$ and $\hat{\mathbf{w}}_\mathrm{p}$ as in \eqref{eq1} and \eqref{eq2} based on $\mathcal{K}_\mathrm{p}$.
	
	{\bf Subsampling:} Run Algorithm \ref{alg1} with expected subsample size $r$ to attain the subsample set $\mathcal{K}$, then attain $\hat{\boldsymbol{\theta}}_{uni}$ by solving \eqref{eq4} based on the subsample set $\mathcal{K}$. 
	
	{\bf One-step Estimation:} Get $\tilde{\boldsymbol{\theta}}$ according to \eqref{eq6}.
	
	{\bf Approximation of Function and Model:} For $\boldsymbol{V}$, $\boldsymbol{\Phi}$, and $\mathcal{T}(\boldsymbol{\theta}_0,\boldsymbol{\gamma}_0,\mathbf{w}_0)$ in Theorem \ref{thm1}, replace their population versions with corresponding sample estimates $\tilde{\boldsymbol{V}}$, $\tilde{\boldsymbol{\Phi}}$, and $\mathcal{T}(\tilde{\boldsymbol{\theta}},\hat{\boldsymbol{\gamma}}_\mathrm{p},\hat{\mathbf{w}}_\mathrm{p})$ based on $\mathcal{K}_\mathrm{p}$ and substitute the unknown true parameters $\boldsymbol{\theta}_0$, $\boldsymbol{\gamma}_0$, and $\mathbf{w}_0$ with their estimations $\tilde{\boldsymbol{\theta}}$, $\hat{\boldsymbol{\gamma}}_\mathrm{p}$, and $\hat{\mathbf{w}}_\mathrm{p}$. Estimate $h(\cdot)$ by substituting sample estimates $\tilde{\boldsymbol{\Phi}}$ and $\mathcal{T}(\tilde{\boldsymbol{\theta}},\hat{\boldsymbol{\gamma}}_\mathrm{p},\hat{\mathbf{w}}_\mathrm{p})$, denoted as $\tilde{h}(\cdot)=(\tilde{h}_1(\cdot),\cdots,\tilde{h}_d(\cdot))$.
	
	{\bf Monte Carlo Simulation:} Generate $r_{m}$ random samples $\{\tilde{\boldsymbol{U}}^{(1)},\cdots,\tilde{\boldsymbol{U}}^{(r_m)}\}$ from $\mathbb{N}(\mathbf{0}_{d^2+2d},\tilde{\boldsymbol{V}})$. Let $\tilde{h}_{L,j}$ and $\tilde{h}_{U,j}$ be the lower and upper $\alpha/2$ quantiles of $\{\tilde{h}_j(\tilde{\boldsymbol{U}}^{(i)})\}_{i=1}^{r_m}$ for $j=1,\cdots,d$, respectively.
	
	{\bf Confidence Interval Construction:} For $j=1,\cdots,d$, construct the confidence interval of the $j$-th component of $\boldsymbol{\theta}_0$ with confidence level $1-\alpha$ as $[\tilde{\theta}-\tilde{h}_{U,j}/\min\{\sqrt{n},r\},\tilde{\theta}-\tilde{h}_{L,j}/\min\{\sqrt{n},r\}]$.
	
	\caption{Confidence interval construction via the DVS estimator}\label{alg2}
\end{algorithm}
In conclusion, regardless of the selection of $r$ relative to $\sqrt{n}$, \eqref{eq10} can be utilized to construct the confidence interval by choosing an appropriate value for $r_\mathrm{p}$, as discussed previously. Consequently, the confidence interval for the proposed estimator is derived using the Monte Carlo approximation of the asymptotic distribution presented in Theorem~\ref{thm1}, as outlined in Algorithm \ref{alg2}.

It is noteworthy that in the Monte Carlo step, $r_m$ refers to the number of random samples generated to estimate the quantiles used for constructing the confidence interval. These samples are drawn from $\tilde{h}(\tilde{\boldsymbol{U}})$, where $\tilde{\boldsymbol{U}}$ follows a normal distribution with mean $\mathbf{0}_{d^2+2d}$ and the covariance matrix $\tilde{\boldsymbol{V}}$ is defined in Algorithm \ref{alg2}.

\subsection{Inference via the multi-step estimator}\label{sec3.2}
In the previous subsection, the asymptotic properties of the DVS estimator is provided, affording practitioners enhanced flexibility in parameter estimation and confidence interval construction across varying selections of $r$ and $r_\mathrm{p}$ under large scale data contexts. Specifically, if the computational efficiency is required, one can choose a small $r$ and the corresponding small $r_\mathrm{p}$ to attain the DVS estimator and achieve fast inference. However, if high accuracy is desired, Algorithm \ref{alg2} may remain computationally intensive while attaining $\hat{\boldsymbol{\theta}}_{uni}$, though reducing much time compared to the full-data-based estimator. 

Inspired by \citet{ning2017general}, A naive idea is to construct a one-step estimator directly based on the initial estimator $\hat{\boldsymbol{\theta}}_{\mathrm{p}}$ provided in \eqref{eq1}, rather than the estimator obtained via $\hat{\boldsymbol{\theta}}_{uni}$. Note that in the previous subsection, $\hat{\boldsymbol{\theta}}_{uni}$ can also be viewed as a one-step estimation based on $\hat{\boldsymbol{\theta}}_\mathrm{p}$ and the subsample $\mathcal{K}$, and thus the DVS estimator can be regarded as a two-step estimator. Therefore, we propose a multi-step estimator as follows.
\begin{definition}\label{def2}
	Denote by $\hat{\boldsymbol{\theta}}_{(0)}=\hat{\boldsymbol{\theta}}_\mathrm{p}$ the initial Lasso estimator \eqref{eq1} based on the pilot subsample $\mathcal{K}_\mathrm{p}$ . For $\ell=1,\cdots$, we define
	\begin{equation}
		\begin{aligned}\label{eq13}
			\hat{\boldsymbol{\theta}}_{(\ell)}=\hat{\boldsymbol{\theta}}_{(\ell-1)}-\hat{\boldsymbol{\Phi}}_\mathrm{p}^{-1}\frac{1}{n}\sum_{i=1}^{n}\left(b'(\hat{\boldsymbol{\theta}}_{(\ell-1)}^\top \boldsymbol{z}_{i}+\hat{\boldsymbol{\gamma}}_\mathrm{p}^\top\boldsymbol{u}_{i})-y_i\right)\left(\boldsymbol{z}_i-\hat{\mathbf{w}}_\mathrm{p}\boldsymbol{u}_i\right),
		\end{aligned}
	\end{equation}
\end{definition}
where $\hat{\boldsymbol{\Phi}}_\mathrm{p}=\nabla_{\boldsymbol{\theta}}S^*(\hat{\boldsymbol{\theta}}_{\mathrm{p}},\hat{\boldsymbol{\gamma}}_\mathrm{p},\hat{\mathbf{w}}_\mathrm{p})=r_\mathrm{p}^{-1}\sum_{i\in\mathcal{K}_\mathrm{p}}b''(\hat{\boldsymbol{\beta}}_{\mathrm{p}}^\top \boldsymbol{x}_{i})\left(\boldsymbol{z}_i-\hat{\mathbf{w}}_\mathrm{p}\boldsymbol{u}_i\right)\boldsymbol{z}_i^\top$ is an estimator of $\boldsymbol{\Phi}$ based on the pilot subsample $\mathcal{K}_\mathrm{p}$.

Note that only $\hat{\boldsymbol{\theta}}_{(\ell-1)}^\top \boldsymbol{z}_{i}$ and $b'(\cdot)$ should be updated in each iteration. Therefore, compared to the one-step estimator, the additional computational cost of $\hat{\boldsymbol{\theta}}_{(\ell)}$ is $O(nd(\ell-1))$, which is small. Compared with the DVS estimator, the one-step estimator $\hat{\boldsymbol{\theta}}_{(1)}$ saves the time of calculating $\hat{\boldsymbol{\theta}}_{uni}$. The Bahadur representation of $\hat{\boldsymbol{\theta}}_{(\ell)}$ is established as follows.

\begin{theorem}\label{thm3}
	Under Assumptions (A.1)-(A.5), for each $\ell=1,\cdots$, we have
	\begin{equation}\label{eq14}
		\begin{aligned}
			&\sqrt{n}\{(\hat{\boldsymbol{\theta}}_{(\ell)}-\boldsymbol{\theta}_0)+\boldsymbol{\Phi}^{-1}S(\boldsymbol{\theta}_0,\boldsymbol{\gamma}_0,\mathbf{w}_0)\}\\
			=&O_P(n^{1/2}(s_1\vee \sqrt{s_1s_2})r_\mathrm{p}^{-1}\log p+n^{1/2}(s_1^3\vee s_1^2s_2)(\log p/r_\mathrm{p})^{3/2}+s_1(\log p/r_\mathrm{p})^{1/2}+s_2\log p/r_\mathrm{p}^{1/2}).
		\end{aligned}
	\end{equation}
	Specifically, if $s_1\log p/n^{1/2}=O(1)$ and $s_1(\log p/r_\mathrm{p})^{1/2}=o(1)$, the error in \eqref{eq14} can be written as $\sqrt{n}\{(\hat{\boldsymbol{\theta}}_{(\ell)}-\boldsymbol{\theta}_0)+\boldsymbol{\Phi}^{-1}S(\boldsymbol{\theta}_0,\boldsymbol{\gamma}_0,\mathbf{w}_0)\}=\sqrt{n}(\eta+\iota_{(\ell)})$, where $\eta=O_P((s_1^2\vee s_1s_2)(\log p/r_\mathrm{p})^{3/2}+(s_1s_2)^{1/2}\log p/r_\mathrm{p}+s_1\log p/(nr_\mathrm{p}^{1/2}))$ does not depend on $\ell$, and $\iota_{(1)}=O_P(s_1^2\log p/r_\mathrm{p})$, $\iota_{(\ell)}=O_P((\eta+n^{-1/2})^2+(\eta+n^{-1/2})(\sqrt{\log p/r_\mathrm{p}}+s_1^2\log p/r_\mathrm{p}))+(\eta+n^{-1/2}+\sqrt{\log p/r_\mathrm{p}}+s_1^2\log p/r_\mathrm{p})\iota_{(\ell-1)}+\iota_{(\ell-1)}^2)$ for $\ell\geq 2$.
	
	Furthermore, under Assumptions (A.1)-(A.5), if $r_\mathrm{p}$ satisfies 
	\begin{equation}\label{eq15}
		\begin{aligned}
			r_\mathrm{p}/(\max\{n^{1/2}(s_1\vee \sqrt{s_1s_2})\log p,n^{1/3}(s_1^2\vee s_1^{\frac{4}{3}}s_2^{\frac{2}{3}})\log p,s_2^2\log^2 p\}\to\infty,
		\end{aligned}
	\end{equation}
	we have $\sqrt{n}(\hat{\boldsymbol{\theta}}_{(\ell)}-\boldsymbol{\theta}_0)\mathop{\to}\limits^{d} N(0,c(\sigma_0)\boldsymbol{\Phi}^{-1})$.
\end{theorem}

From Theorem \ref{thm3}, it follows that under some conditions and if $\iota_{(1)}=o_P(1),\eta+n^{-1/2}+\sqrt{\log p/r_\mathrm{p}}+s_1^2\log p/r_\mathrm{p}=o_P(1)$, the error of the Bahadur representation decreases as $\ell$ increases. The error of the Bahadur representation of $\hat{\boldsymbol{\theta}}_{(\ell)}$ will finally be dominated by $\eta$ and the terms in $\iota_{(\ell)}$ that does not depend on $\iota_{(\ell-1)}$. Therefore, we let the iteration proceed until the relative change in the estimation becomes sufficiently small, that is, when the ratio $\left(\mbox{error}_{(\ell)}-\mbox{error}_{(\ell-1)}\right)/\mbox{error}_{(\ell-1)}$ falls below a predefined threshold. Here, $\mbox{error}_{(\ell)}=\|\hat{\boldsymbol{\theta}}_{(\ell)}-\hat{\boldsymbol{\theta}}_{(\ell-1)}\|$  denotes the stepwise estimation discrepancy between successive iterates. Similar to Theorem \ref{thm2}, Theorem \ref{thm3} also implies the semiparametric efficiency of the multi-step estimator $\hat{\boldsymbol{\theta}}_{(\ell)}$ ($\ell\geq 1$). We summarize the procedure for obtaining the multi-step estimator and constructing confidence intervals in Algorithm \ref{alg3}.
\begin{algorithm}[!htbp]
	\SetAlgoLined
	{\bf Subsampling:} Run Algorithm \ref{alg1} with expected subsample size $r_\mathrm{p}$ to get a subsample $\mathcal{K}_\mathrm{p}$, and attain  $\hat{\boldsymbol{\beta}}_{\mathrm{p}}=(\hat{\boldsymbol{\theta}}_{\mathrm{p}}^\top,\hat{\boldsymbol{\gamma}}_{\mathrm{p}}^\top)^\top$ and $\hat{\mathbf{w}}_\mathrm{p}$ via \eqref{eq1}-\eqref{eq2}.
	
	{\bf Multi-step Estimation:} Set $\hat{\boldsymbol{\theta}}_{(0)}=\hat{\boldsymbol{\theta}}_{\mathrm{p}}$, $\mbox{error}_{(0)}=1$.
	
	\For{$\ell=1,\cdots,\text{maxiter}$}{
		Calculate the multi-step estimator $\hat{\boldsymbol{\theta}}_{(\ell)}$ according to \eqref{eq13}. Denote $\mbox{error}_{(\ell)}=\|\hat{\boldsymbol{\theta}}_{(\ell)}-\hat{\boldsymbol{\theta}}_{(\ell-1)}\|$.
		
		\textbf{Break if} $\left(\mbox{error}_{(\ell)}-\mbox{error}_{(\ell-1)}\right)/\mbox{error}_{(\ell-1)}<10^{-3}$. Let $\hat{\boldsymbol{\theta}}_{(end)}=\hat{\boldsymbol{\theta}}_{(\ell)}$.
	}
	
	{\bf Confidence interval construction:} Denote the $(j,j)$-th entry of $\hat{\boldsymbol{\Phi}}_\mathrm{p}^{-1}$ as $\nu_{j,j}$. For $j=1,\cdots,d$, construct the confidence interval of the $j$-th component of $\boldsymbol{\theta}_0$ with confidence level $1-\alpha$ as $[\hat{\boldsymbol{\theta}}_{(end)}-\nu_{j,j}\rho_{1-\frac{\alpha}{2}}/\sqrt{n},\hat{\boldsymbol{\theta}}_{(end)}+\nu_{j,j}\rho_{1-\frac{\alpha}{2}}/\sqrt{n}]$, where $\rho_{1-\frac{\alpha}{2}}$ is the $1-\frac{\alpha}{2}$-th quantile of the standard normal distribution.
	\caption{Confidence interval construction via the multi-step estimator}\label{alg3}
\end{algorithm}

We now turn to the choice of $r_\mathrm{p}$. Equation \eqref{eq14} indicates that the remainder term in the Bahadur representation is asymptotically negligible, and valid confidence intervals may be constructed via Algorithm \ref{alg3}, provided that $r_{\mathrm{p}}\gg \max\{n^{1/2}(s_1\vee(s_1s_2)^{1/2})\log p,n^{1/3}(s_1^2\vee s_1^{4/3}s_2^{2/3})\log p,s_2^2\log^2 p\}$. In practice, we recommend choosing $r_{\mathrm{p}}= n/C$ for some constant $C$ (e.g., $C=5$) to ensure computational efficiency while preserving inferential validity.


\section{High-dimensional simultaneous inference}\label{sec4}
This section focuses on a high-dimensional parameter vector $\boldsymbol{\theta}\in \mathbb{R}^d$ of scientific interest, where the dimension $d$ may grow with the sample size. This setting arises in many modern applications, such as genomics, where one jointly tests the effects of many genes, or econometrics, where the goal is to assess a large set of policy variables. In such cases, standard coordinate-wise confidence intervals fail to provide a coherent statistical assessment, as the overall confidence level for jointly covering all parameters deteriorates rapidly due to multiple testing-a “curse of dimensionality’’ for inference. Beyond this statistical challenge, the computational cost can become intractable as both $n$ and $p$ grow. Therefore, we develop a scalable framework for simultaneous inference on high-dimensional target parameters by deriving a uniform Bahadur representation for our debiased estimator based on the decorrelated score function, and then applying the multiplier bootstrap for high-dimensional simultaneous inference. This approach enables the construction of simultaneous confidence regions that remain both theoretically valid and computationally efficient even as $d$ diverges.

Motivated by the debiased Lasso \citep{javanmard2014confidence,van2014asymptotically,zhang2014confidence} and the decorrelated score function \citep{ning2017general,cheng2022regularized,fang2020test}, the crucial step for attaining the unbiased parameter is to construct an approximation of the inverse Hessian matrix, which serves a similar role to the estimation of $\boldsymbol{\Phi}^{-1}$ in \eqref{eq6}. When the dimension of parameters of interest is diverge, that is, $d$ is also large, $(\nabla_{\boldsymbol{\theta}}S^*(\hat{\boldsymbol{\theta}}_{uni},\hat{\boldsymbol{\gamma}}_\mathrm{p},\hat{\mathbf{w}}_\mathrm{p}))^{-1}$ in \eqref{eq6} or $(\nabla_{\boldsymbol{\theta}}S^*(\hat{\boldsymbol{\theta}}_{\mathrm{p}},\hat{\boldsymbol{\gamma}}_\mathrm{p},\hat{\mathbf{w}}_\mathrm{p}))^{-1}$ in \eqref{eq13} may not be a good approximation of $\boldsymbol{\Phi}^{-1}$. Therefore, we adopt the approach in \citet{cai2011constrained,yan2023confidence} to attain the estimator of $\boldsymbol{\Phi}^{-1}$, denoted as $\hat{\mathbf{G}}$, via the following convex program:
\begin{align}\label{eq16}
	\min_{\mathbf{G}\in\mathbb{R}^{d\times d}}\|\mathbf{G}\|_{1,1},\quad \text{s.t.}\quad \|\mathbf{G}\check{\boldsymbol{\Phi}}_\mathrm{p}-\mathbf{I}\|_{\infty}\leq\gamma_{n},
\end{align}
where $\check{\boldsymbol{\Phi}}_\mathrm{p}=({1}/{r_\mathrm{p}})\sum_{i\in\mathcal{K}_\mathrm{p}}b''(\hat{\boldsymbol{\beta}}_{\mathrm{p}}^\top \boldsymbol{x}_{i})\left(\boldsymbol{z}_i-\hat{\mathbf{w}}_\mathrm{p}\boldsymbol{u}_i\right)\left(\boldsymbol{z}_i-\hat{\mathbf{w}}_\mathrm{p}\boldsymbol{u}_i\right)^\top$ is a symmetric covariance-type estimator and is still a sample version of $\boldsymbol{\Phi}$ based on the subsample $\mathcal{K}_\mathrm{p}$ and utilizing these initial estimators $\hat{\boldsymbol{\beta}}_{\mathrm{p}},\hat{\mathbf{w}}_\mathrm{p}$, and $\gamma_n$ is a predetermined tuning parameter. While $\boldsymbol{\Phi}$ is symmetric as illustrated in Section \ref{sec3.1}, $\hat{\mathbf{G}}=(\hat{g}_{i,j})_{1\leq i,j\leq d}$ is not symmetric typically. To address this problem, one can define $\tilde{\mathbf{G}}=(\tilde{g}_{i,j})_{1\leq i,j\leq d}$, where $\tilde{g}_{i,j}=\hat{g}_{i,j}I(|\hat{g}_{i,j}|\leq|\hat{g}_{j,i}|)+\hat{g}_{j,i}I(|\hat{g}_{j,i}|<|\hat{g}_{i,j}|)$. Therefore, we assume that $\hat{\mathbf{G}}$ is symmetric without loss of generality in our work. Then we propose the debiased estimator $\check{\boldsymbol{\theta}}$ as follows: 
\begin{equation}
	\begin{aligned}\label{eq17}
		\check{\boldsymbol{\theta}}=&\hat{\boldsymbol{\theta}}_{\mathrm{p}}-\hat{\mathbf{G}}S(\hat{\boldsymbol{\theta}}_{\mathrm{p}},\hat{\boldsymbol{\gamma}}_\mathrm{p},\hat{\mathbf{w}}_\mathrm{p})=\hat{\boldsymbol{\theta}}_{\mathrm{p}}-\hat{\mathbf{G}}\frac{1}{n}\sum_{i=1}^{n}\left(b'(\hat{\boldsymbol{\beta}}_{\mathrm{p}}^\top \boldsymbol{x}_{i})-y_i\right)\left(\boldsymbol{z}_i-\hat{\mathbf{w}}_\mathrm{p}\boldsymbol{u}_i\right).
	\end{aligned}
\end{equation}
To further develop the asymptotic behavior of $\check{\boldsymbol{\theta}}$, another assumption is needed.\\
\noindent
(B.1) Assume that $\boldsymbol{\Phi}^{-1}=(\mathbf{g}_{1},\cdots,\mathbf{g}_{d})^\top=(g_{i,j})_{1\leq i,j\leq d}$ is row-wisely sparse, i.e., $\max_{1\leq i\leq d}\\\sum_{j=1}^{d}|g_{ij}|^q<C_g$ for some $q\in[0,1)$ and positive constant $C_g$. The tuning parameter $\gamma_n$ in \eqref{eq16} is set as $\gamma_n\asymp (s_1\vee \sqrt{s_2})(\log p/r_\mathrm{p})^{1/2}$.

This assumption naturally holds for fixed $d$ and requires $\boldsymbol{\Phi}$ to be sparse in the sense of the $\ell_q$ norm of matrix row space if $d$ diverges. Similar conditions can be seen in \citet{cai2011constrained,yan2023confidence,ning2017general,cai2025statistical}. It is the most common matrix sparsity assumption under the case of $q=0$. More specifically, our sparsity assumptions on $\boldsymbol{\Phi}^{-1}$ (Assumption (B.1)) and $\mathbf{w}_0$ (Assumption (A.2)) can be weaker than those sparsity functions on debiased lasso \citep{van2014asymptotically} as they can be deduced from the sparsity of $(\mathbb{E}[b''(\boldsymbol{\beta}_0^\top\boldsymbol{x})\boldsymbol{x}\boldsymbol{x}^\top])^{-1}$, as illustrated in Remark \ref{remark3.1}. Based on Assumption (B.1), the estimation error $\|\hat{\mathbf{G}}-\boldsymbol{\Phi}^{-1}\|_{L_1}$ can then be established. In the following theorem, we present the uniform Bahadur representation for $\check{\boldsymbol{\theta}}$ under the case of $q=0$ for convenience, and the whole analysis is exhibited in the supplementary materials.

\begin{theorem}\label{thm4}
	Suppose Assumptions (A.1)-(A.5) and (B.1) hold under the case of $q=0$, if $n$ satisfies $(s_1\vee s_2)\sqrt{\log p/n}=o(1)$, then 
	\begin{align*}
		\|\sqrt{n}\{(\check{\boldsymbol{\theta}}-\boldsymbol{\theta}_0)+\boldsymbol{\Phi}^{-1}S(\boldsymbol{\theta}_0,\boldsymbol{\gamma}_0,\mathbf{w}_0)\}\|_\infty=O_P\left(n^{1/2}s_1(s_1\vee \sqrt{s_2})(\log p/r_\mathrm{p})+s_2\log p/r_\mathrm{p}^{1/2}\right).
	\end{align*}
\end{theorem}

Theorem \ref{thm4} illustrates the uniform Bahadur representation for $\check{\boldsymbol{\theta}}$, implying that $\sqrt{n}(\check{\boldsymbol{\theta}}-\boldsymbol{\theta}_0)$ can be expressed by a linear transform of $S(\boldsymbol{\theta}_0,\boldsymbol{\gamma}_0,\mathbf{w}_0)$ up to some ignorable terms. It also exhibits its semiparametric efficiency. If we choose $r_\mathrm{p}=n$, i.e., all samples are utilized, then $s_1(s_1\vee\sqrt{s_2})\log p/n^{1/2}=o(1)$ and $s_2^2\log^2(p)/n=o(1)$ is required for the smallest sample size $n$ to make the residuals to be ignorable. Moreover, the subsampling technique is employed to alleviate computational cost if $n$ is large.

Now we aim to construct simultaneous confidence intervals for $\boldsymbol{\theta}$. From Theorem \ref{thm4}, it is natural to derive that the asymptotic distribution of $\|\sqrt{n}(\check{\boldsymbol{\theta}}-\boldsymbol{\theta})\|_\infty$ can be approximated by the maximum norm of a Gaussian random vector with mean $\mathbf{0}$ and variance $c(\sigma_0)\boldsymbol{\Phi}^{-1}$. This is formally stated in the following proposition.
\begin{proposition}\label{prop3}
	Suppose that Assumptions (A.1)-(A.5) and (B.1) hold under the case of $q=0$, if $(s_1\vee s_2)\sqrt{\log p/n}=o(1)$, $n^{1/2}s_1(s_1\vee \sqrt{s_2})(\log p/r_\mathrm{p})=o(1)$, $s_2\log p/r_\mathrm{p}^{1/2}=o(1)$, $\log^7 (dn)/n=o(1)$, then there exists a $d$-dimensional Gaussian random vector $\boldsymbol{Q}_n$ with mean $\mathbf{0}$ and variance $c(\sigma_0)\boldsymbol{\Phi}^{-1}$ such that 
	$$\sup_{t\in\mathbb{R}}\left|\mathbb{P}\left(\|\sqrt{n}(\check{\boldsymbol{\theta}}-\boldsymbol{\theta})\|_\infty\leq t\right)-\mathbb{P}(\|\boldsymbol{Q}_n\|_\infty\leq t)\right|=o(1).$$
\end{proposition}

However, when $d$ is large, it is still unclear whether the maximum norm of $\boldsymbol{Q}_n$ can be reliably estimated by sampling from the $d$-dimensional Gaussian distribution with an estimated covariance matrix. In contrast, we adopt the multiplier bootstrap procedure. Let $e_i\sim \mathbb{N}(0,1),i=1,\cdots,n$ be independent of the data $\mathcal{D}=\{\boldsymbol{x}_i,y_i\}_{i=1}^{n}$, and consider $$\hat{\boldsymbol{Q}}_n^e=\sqrt{n}\hat{\mathbf{G}}\frac{1}{n}\sum_{i=1}^{n}e_i\left(b'(\hat{\boldsymbol{\beta}}_\mathrm{p}^\top\boldsymbol{x}_i)-y_i\right)(\boldsymbol{z}_i-\hat{\mathbf{w}}_\mathrm{p}\boldsymbol{u}_i).$$ Denote the $\alpha$-quantile of $\|\hat{\boldsymbol{Q}}_n^e\|_\infty$ by $c_{n}(\alpha)=\inf_{t\in\mathbb{R}}\{t\in\mathbb{R}:\mathbb{P}_e(\|\hat{\boldsymbol{Q}}_n^e\|_\infty\leq t)\geq\alpha\},$
where $\mathbb{P}_e(\mathcal{A})$ represents the probability of the event $\mathcal{A}$ with respect to $e_1,\cdots,e_n$. The validity of the bootstrap approach is established by the following theorem.

\begin{theorem}\label{thm5}
	Under the conditions of Proposition \ref{prop3}, if $\log(p)\log\log(p)\{s_1^2\log p/r_\mathrm{p}+(s_1\vee s_2)(\log p/r_\mathrm{p})^{1/2}\}=o(1)$, we have $$\sup_{t\in\mathbb{R}}\left|\mathbb{P}\left(\|\hat{\boldsymbol{Q}}_n^e\|_\infty\leq t|\mathcal{D}\right)-\mathbb{P}\left(\|\boldsymbol{Q}_n\|_\infty\leq t\right)\right|=o(1).$$
\end{theorem}

Following Proposition \ref{prop3} and Theorem \ref{thm5}, denote $\check{\boldsymbol{\theta}}=(\check{\theta}_1,\cdots,\check{\theta}_d)$, the simultaneous $\alpha$-level confidence interval for $\boldsymbol{\theta}=(\theta_1,\cdots,\theta_d)^\top$ can be constructed as follows:
$$\left(\check{\theta}_i-n^{-1/2}c_{n}(1-\alpha),\check{\theta}_i+n^{-1/2}c_{n}(1-\alpha)\right),i=1,\cdots,d.$$
Moreover, we propose a studentized statistic $\|\sqrt{n}\hat{g}_{j,j}^{-1/2}(\check{\boldsymbol{\theta}}-\boldsymbol{\theta})_j\|_\infty$ for $j=1,\cdots,d$ by leveraging the idea in \citet{zhang2017simultaneous,cai2025statistical}, which can also be considered to attain confidence intervals with varying length, where $\hat{g}_{j,j}$ is the $(j,j)$-th element of $\hat{\mathbf{G}}$. Denote $\boldsymbol{S}=\mbox{diag}(\boldsymbol{\Phi}^{-1})$ and $\hat{\boldsymbol{S}}=\mbox{diag}(\hat{\mathbf{G}})$. Let $$\hat{\boldsymbol{Q}}_n^{e,stu}=\sqrt{n}\hat{\boldsymbol{S}}^{-1/2}\hat{\mathbf{G}}\frac{1}{n}\sum_{i=1}^{n}e_i\left(b'(\hat{\boldsymbol{\beta}}_\mathrm{p}^\top\boldsymbol{x}_i)-y_i\right)(\boldsymbol{z}_i-\hat{\mathbf{w}}_\mathrm{p}\boldsymbol{u}_i),$$ and denote the $\alpha$-quantile of $\|\hat{\boldsymbol{Q}}_n^{e,stu}\|_\infty$ by $c_{n}^{stu}(\alpha)=\inf_{t\in\mathbb{R}}\{t\in\mathbb{R}:\mathbb{P}_e(\|\hat{\boldsymbol{Q}}_n^{e,stu}\|_\infty\leq t)\geq\alpha\}.$
The validity of the studentized bootstrap approach is established by the following theorem.

\begin{theorem}\label{thm6}
	Under the conditions of Proposition \ref{prop3}, there exists a $d$-dimensional Gaussian random vector $\boldsymbol{Q}_n^{stu}$ with mean $\mathbf{0}$ and variance $c(\sigma_0)\boldsymbol{S}^{-1/2}\boldsymbol{\Phi}^{-1}\boldsymbol{S}^{-1/2}$ such that $$\sup_{t\in\mathbb{R}}\left|\mathbb{P}\left(\|\sqrt{n}\hat{\boldsymbol{S}}^{-1/2}(\check{\boldsymbol{\theta}}-\boldsymbol{\theta})\|_\infty\leq t\right)-\mathbb{P}(\|\boldsymbol{Q}_n^{stu}\|_\infty\leq t)\right|=o(1).$$
	Furthermore, under the conditions of Theorem \ref{thm5},
	$$\sup_{t\in\mathbb{R}}\left|\mathbb{P}\left(\|\hat{\boldsymbol{Q}}_n^{e,stu}\|_\infty\leq t\right)-\mathbb{P}(\|\boldsymbol{Q}_n^{stu}\|_\infty\leq t)\right|=o(1).$$
\end{theorem}

Following Theorem \ref{thm6}, the simultaneous $\alpha$-level confidence interval with varying length for $\boldsymbol{\theta}=(\theta_1,\cdots,\theta_d)^\top$ can be constructed as $$(\check{\theta}_i-n^{-1/2}\hat{g}_{i,i}^{1/2}c_{n}^{stu}(1-\alpha),\check{\theta}_i+n^{-1/2}\hat{g}_{i,i}^{1/2}c_{n}^{stu}(1-\alpha)), \quad i=1,\cdots,d.$$ 
\begin{algorithm}[htbp]
	\SetAlgoLined
	{\bf Subsampling and initial estimation:} Run Algorithm \ref{alg1} with expected subsample size $r_\mathrm{p}$ to get a initial subsample set $\mathcal{K}_\mathrm{p}$, and obtain $\hat{\boldsymbol{\beta}}_{\mathrm{p}}$ and $\hat{\mathbf{w}}_\mathrm{p}$ via \eqref{eq1} and \eqref{eq2}. Compute $\check{\boldsymbol{\Phi}}_\mathrm{p}=\frac{1}{r_\mathrm{p}}\sum_{i\in\mathcal{K}_\mathrm{p}}b''(\hat{\boldsymbol{\beta}}_{\mathrm{p}}^\top \boldsymbol{x}_{i})\left(\boldsymbol{z}_i-\hat{\mathbf{w}}_\mathrm{p}\boldsymbol{u}_i\right)\left(\boldsymbol{z}_i-\hat{\mathbf{w}}_\mathrm{p}\boldsymbol{u}_i\right)^\top$. Solve \eqref{eq16} to attain $\hat{\mathbf{G}}$, let $\hat{\boldsymbol{S}}=\mbox{diag}(\hat{\mathbf{G}})=\mbox{diag}(\hat{g}_{1,1},\cdots,\hat{g}_{d,d})$, and attain $\check{\boldsymbol{\theta}}=(\check{\theta}_1,\cdots,\check{\theta}_d)^\top$ via \eqref{eq17}.
	
	\For{$k=1,\cdots,B$}{
		Generate standard normal random variables $e_{ki},i=1,\cdots,n$.
		
		Calculate $\hat{\boldsymbol{Q}}_{n,k}^{e,stu}=\sqrt{n}\hat{\boldsymbol{S}}^{-1/2}\hat{\mathbf{G}}\frac{1}{n}\sum_{i=1}^{n}e_{ki}\left(b'(\hat{\boldsymbol{\beta}}_\mathrm{p}^\top\boldsymbol{x}_i)-y_i\right)(\boldsymbol{z}_i-\hat{\mathbf{w}}_\mathrm{p}\boldsymbol{u}_i)$
	}
	
	{\bf Confidence interval construction:} Find the $\alpha$-quantile of $\{\|\hat{\boldsymbol{Q}}_{n,k}^{e,stu}\|_\infty\}_{k=1}^{B}$, represented as $c_{n}^{stu}(1-\alpha)$. The simultaneous $\alpha$-level confidence interval for $\theta_i$ is $\left(\check{\theta}_i-n^{-1/2}\hat{g}_{i,i}^{1/2}c_{n}^{stu}(1-\alpha),\check{\theta}_i+n^{-1/2}\hat{g}_{i,i}^{1/2}c_{n}^{stu}(1-\alpha)\right),i=1,\cdots,d.$
	\caption{Simultaneous confidence intervals through Bootstrap}\label{alg4}
\end{algorithm}
The whole steps to attain $\check{\boldsymbol{\theta}}$ and construct simultaneous confidence intervals with varying length are exhibited in Algorithm \ref{alg4}. The corresponding procedure for the non-studentized version is structurally similar and thus omitted for brevity. Note that conditions in Proposition \ref{prop3} implies that $r_\mathrm{p}$ should be selected such that $r_{\mathrm{p}}\gg \max\{n^{1/2}s_1(s_1\vee s_2^{1/2})\log p,s_2^2\log^2 p\}$. In implementation, $r_{\mathrm{p}}=n/C$ for some constant $C$ (e.g., $C=5$) suffices similar to the multi-step estimator in Section \ref{sec3.2}.
\section{Numerical experiments}\label{sec5}
In this section, we assess the finite-sample performance of our three estimators under both linear and logistic models. We repeated the simulations $500$ times, all executed in \texttt{R} on a system equipped with an Intel Core i9-12900KF CPU and 64 GB of DDR5 RAM. Due to space limits, the logistic regression results are placed in the supplementary material.

We conduct numerical experiments as follows:\\
(i) For inference on individual coefficients, we compare our DVS estimator $\tilde{\boldsymbol{\theta}}$ and the multi-step estimator $\hat{\boldsymbol{\theta}}_{(end)}$ with several subsampled decorrelated-score estimators from \citet{shan2024optimal}, including the A- and L-optimal weighted and unweighted versions $\check{\boldsymbol{\theta}}_{dw}^{A}$, $\check{\boldsymbol{\theta}}_{dw}^{L}$, $\check{\boldsymbol{\theta}}_{duw}^{A}$, $\check{\boldsymbol{\theta}}_{duw}^{L}$, as well as the uniform subsampling version of subsampled decorrelated score estimator $\check{\boldsymbol{\theta}}_{duni}$. For reference, we also report the computation time of the full-data estimator of \citet{ning2017general}.\\
(ii) For simultaneous inference, we compare our debiased estimator $\check{\boldsymbol{\theta}}$ with nodewise-regression-based debiased estimators from \citet{van2014asymptotically}, restricting nodewise regression to the parameters of interest for computational fairness. Since their method is not directly applicable to high-dimensional simultaneous inference, we additionally adopt \citet{zhang2017simultaneous} for linear models that also uses a multiplier bootstrap. For logistic regression, we treat \citet{van2014asymptotically} as a baseline, again combined with a multiplier bootstrap. Across all settings, we evaluate both non-studentized and studentized procedures, denoted by $\check{\boldsymbol{\theta}}_{ns}$ and $\check{\boldsymbol{\theta}}_{s}$ for our method, and $\hat{\boldsymbol{\theta}}_{ns}$ and $\hat{\boldsymbol{\theta}}_{s}$ for the baseline.

We generate datasets of size 
$n$ from the linear model $y_i=\boldsymbol{\beta}_0^\top \boldsymbol{x}_i+\varepsilon_i$, where
$\boldsymbol{\beta}_0=(\sqrt{3},\sqrt{3},\sqrt{3})^\top\in\mathbb{R}^{p}$, $\boldsymbol{x}_i$'s are i.i.d. covariate vectors, and $\varepsilon_i$ are i.i.d. random errors. The covariates $\boldsymbol{x}_i$ and errors $\varepsilon_i$ are generated from the following four cases: (a) $\boldsymbol{x}\sim N(\mathbf{0},\boldsymbol{\Sigma})$, $\varepsilon\sim N(0,1)$; (b) $\boldsymbol{x}\sim N(\mathbf{0},\boldsymbol{\Sigma})$, $\varepsilon\sim 2t_{5}$; (c) $\boldsymbol{x}\sim t_{10}(\mathbf{0},\boldsymbol{\Sigma})$, $\varepsilon\sim N(0,1)$; (d) $\boldsymbol{x}\sim t_{10}(\mathbf{0},\boldsymbol{\Sigma})$, $\varepsilon\sim 2t_{5}$, where $\boldsymbol{\Sigma}\in\mathbb{R}^{p\times p}$ is a Toeplitz matrix with its $(j,k)$-th component being $0.5^{|j-k|}$.

Since for the linear model, $c(\sigma_0)$ corresponds to the variance of the random error $\varepsilon$. For all the methods, we apply $\sum_{i=1}^{n}\|y_i-\boldsymbol{x}_i^\top\hat{\boldsymbol{\beta}}_\mathrm{p}\|_2^2/(n-\|\hat{\boldsymbol{\beta}}_\mathrm{p}\|_0)$
as an estimator of $c(\sigma_0)$, where $\hat{\boldsymbol{\beta}}_\mathrm{p}$ is a pilot Lasso estimator defined in \eqref{eq1}. This estimator, similar to that in \citet{reid2016study}, helps avoid noise underestimation when $n$ is small.
\subsection{Inference for single coefficient}\label{sec5.1}
We first let the whole data size $n=10^5$ and the pilot subsample size $r_\mathrm{p}=1000$. For the DVS estimator and the subsampling estimators in \citet{shan2024optimal}, we set subsample size $r=r_\mathrm{p}$. The dimension is set to $p=500$, and the parameters of interest are the first $d$ variables with $d=5$ or $d=10$. The confidence level is fixed as $\alpha=0.05$. The comparisons of the MSE, running time, the average coverage probability (ACP) and the average length (AL) of the $95\%$ confidence intervals are delineated in Table \ref{table1}.

\begin{table*}[htbp]
	\caption{MSE, Time (seconds), ACP, and AL for Linear model with $n=10^5$ and $p=500$.}
	\resizebox{\linewidth}{!}{
		\centering
		\begin{tabular}{ccccccccccccccccc}
			\hline
			&$d$& &$\hat{\boldsymbol{\theta}}_{(end)}$&$\tilde{\boldsymbol{\theta}}$&$\check{\boldsymbol{\theta}}_{duw}^{A}$&$\check{\boldsymbol{\theta}}_{duw}^{L}$&$\check{\boldsymbol{\theta}}_{dw}^{A}$&$\check{\boldsymbol{\theta}}_{dw}^{L}$&$\check{\boldsymbol{\theta}}_{duni}$&\cr
			\hline
			\multirow{8}{*}{(a)}&\multirow{4}{*}{5}&MSE&$8.469\times 10^{-5}$&$1.400\times 10^{-4}$&$6.780\times10^{-3}$&$6.781\times10^{-3}$&$7.817\times10^{-3}$&$7.806\times10^{-3}$&$7.812\times10^{-3}$\cr
			& &Time&1.667&1.646&1.522&1.544&1.526&1.562&1.386\cr
			& &ACP&0.948&0.946&0.958&0.945&0.938&0.956&0.952\cr
			& &AL&0.016&0.021&0.140&0.146&0.147&0.154&0.156\cr
			&\multirow{4}{*}{10}&MSE&$1.743\times 10^{-4}$&$3.525\times 10^{-4}$&0.014&0.017&0.015&0.015&0.015\cr
			& &Time&2.667&2.810&2.868&2.567&2.340&2.571&2.033\cr
			& &ACP&0.931&0.952&0.956&0.938&0.950&0.952&0.951\cr
			& &AL&0.016&0.024&0.150&0.154&0.154&0.158&0.158\cr
			\multirow{8}{*}{(b)}&\multirow{4}{*}{5}&MSE&$1.523\times 10^{-4}$&$2.380\times 10^{-4}$&0.011&0.012&0.013&0.013&0.012\cr
			& &Time&1.849&1.776&1.632&1.536&1.623&1.638&1.431\cr
			& &ACP&0.940&0.950&0.944&0.952&0.939&0.952&0.960\cr
			& &AL&0.020&0.027&0.180&0.190&0.190&0.199&0.202\cr
			&\multirow{4}{*}{10}&MSE&$3.280\times 10^{-4}$&$6.871\times 10^{-4}$&0.023&0.024&0.027&0.027&0.027\cr
			& &Time&2.743&3.089&2.798&2.773&2.812&2.714&2.387\cr
			& &ACP&0.934&0.937&0.958&0.954&0.940&0.946&0.944\cr
			& &AL&0.021&0.031&0.192&0.198&0.199&0.206&0.206\cr
			\multirow{8}{*}{(c)}&\multirow{4}{*}{5}&MSE&$7.942\times 10^{-5}$&$1.317\times 10^{-4}$&$4.722\times 10^{-3}$&$4.719\times 10^{-3}$&$5.094\times 10^{-3}$&$6.035\times 10^{-3}$&$6.166\times 10^{-3}$\cr
			& &Time&1.719&1.698&1.534&1.498&1.516&1.528&1.334\cr
			& &ACP&0.929&0.945&0.941&0.952&0.952&0.951&0.960 \cr
			& &AL&0.014&0.020&0.117&0.123&0.127&0.133&0.140\cr
			&\multirow{4}{*}{10}&MSE&$1.539\times 10^{-4}$&$3.473\times 10^{-4}$&0.011&0.011&0.011&0.012&0.012\cr
			& &Time&2.803&3.174&2.481&2.624&2.537&2.594&2.315\cr
			& &ACP&0.930&0.957&0.940&0.944&0.955&0.946&0.954\cr
			& &AL&0.014&0.024&0.126&0.129&0.134&0.138&0.142\cr
			\multirow{8}{*}{(d)}&\multirow{4}{*}{5}&MSE&$1.082\times 10^{-4}$&$2.140\times 10^{-4}$&$7.871\times 10^{-3}$&$8.936\times 10^{-3}$&$9.345\times 10^{-3}$&$9.821\times 10^{-3}$&0.010\cr
			& &Time&1.742&1.712&1.526&1.542&1.552&1.575&1.374\cr
			& &ACP&0.948&0.942&0.937&0.948&0.951&0.945&0.942\cr
			& &AL&0.018&0.026&0.151&0.159&0.164&0.172&0.179\cr
			&\multirow{4}{*}{10}&MSE&$2.442\times 10^{-4}$&$5.810\times 10^{-4}$&0.016&0.019&0.020&0.020&0.021\cr
			& &Time&2.932&3.210&2.952&2.686&2.806&2.576&2.635\cr
			& &ACP&0.936&0.954&0.954& 0.946&0.950&0.956&0.954\cr
			& &AL&0.021&0.031&0.192&0.198&0.199&0.206&0.206\cr
			\hline
		\end{tabular}
	}
	\label{table1}
\end{table*}
Take case (a) and $d=5$ as an example, the full-data-based decorrelated score method takes about $150$ seconds. Thus our methods serve as a computationally economical solution for inferring the target low-dimensional parameters if accuracy is desired. Compared with the subsampling estimators in \citet{shan2024optimal}, it can be seen from Table \ref{table1} that the empirical coverage probabilities of all methods are near $95\%$ and they share similar computational costs. Among these methods, $\hat{\boldsymbol{\theta}}_{(end)}$ achieves the smallest MSEs and the narrowest average lengths, and the DVS estimator is slightly worse than the multi-step estimator, but much better than other baseline methods. Overall, the average length of our method is several times or even nearly ten times shorter than other baseline methods.

Given the special properties of our DVS estimator $\tilde{\boldsymbol{\theta}}$ and multi-step estimator $\hat{\boldsymbol{\theta}}_{(end)}$, we provide a more intuitive comparison. Fixing $n=10^5$, $p=500$, $d=5$, $r_\mathrm{p}=2000$, we examine how the ACPs and ALs vary with the subsample size $r$ for the DVS estimator and the subsampling methods under case (a). For reference, we also include the multi-step estimator with $r_\mathrm{p}=2000$ and the full-data-based method. The results are depicted in Figure \ref{fig1}, while similar trends in other cases are omitted for brevity. 

\begin{figure}[htbp]
	\begin{subfigure}{0.48\textwidth}
		\centering
		\includegraphics[width=\textwidth]{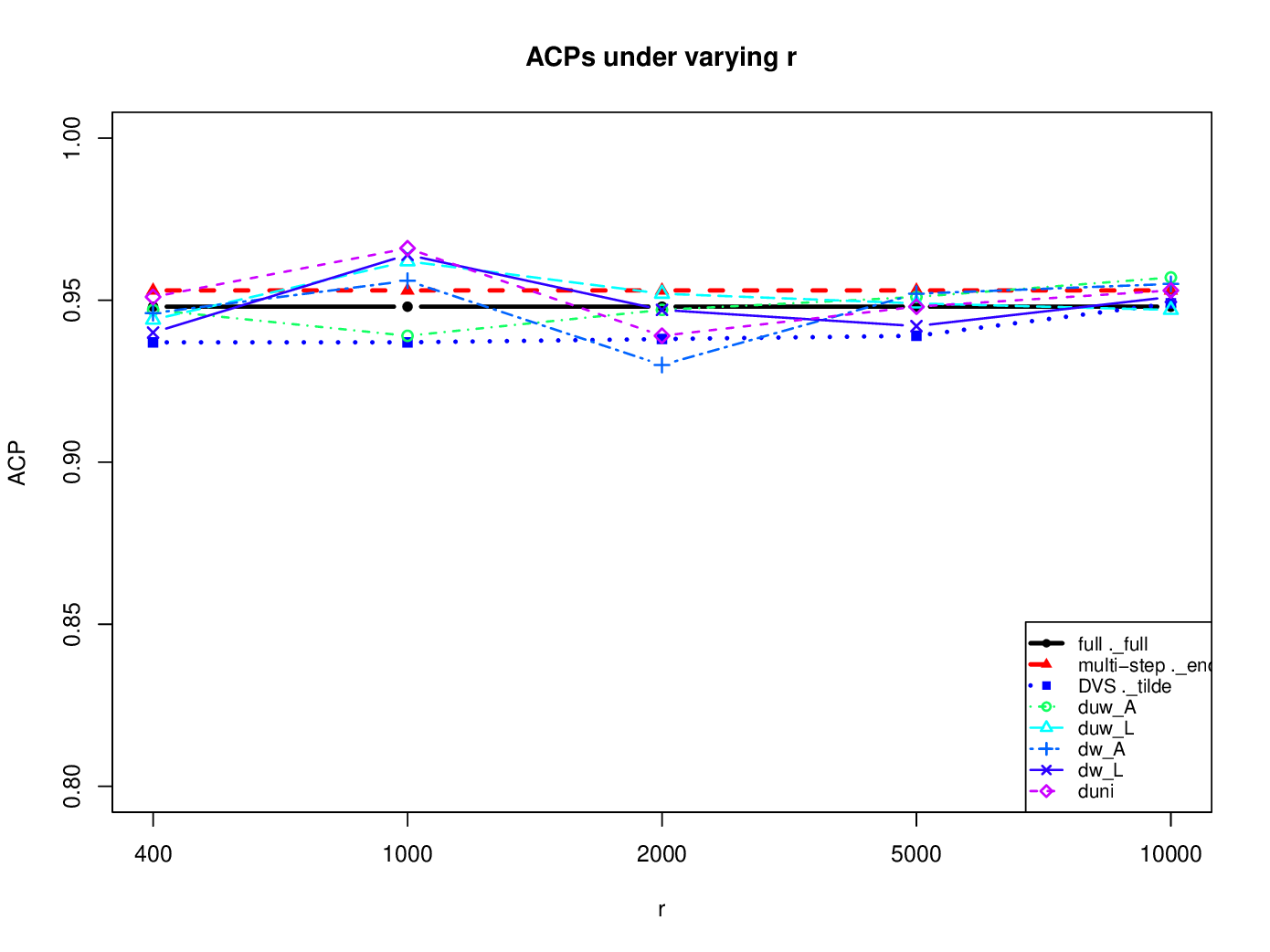}
	\end{subfigure}
	\hfill
	\begin{subfigure}{0.48\textwidth}
		\centering
		\includegraphics[width=\textwidth]{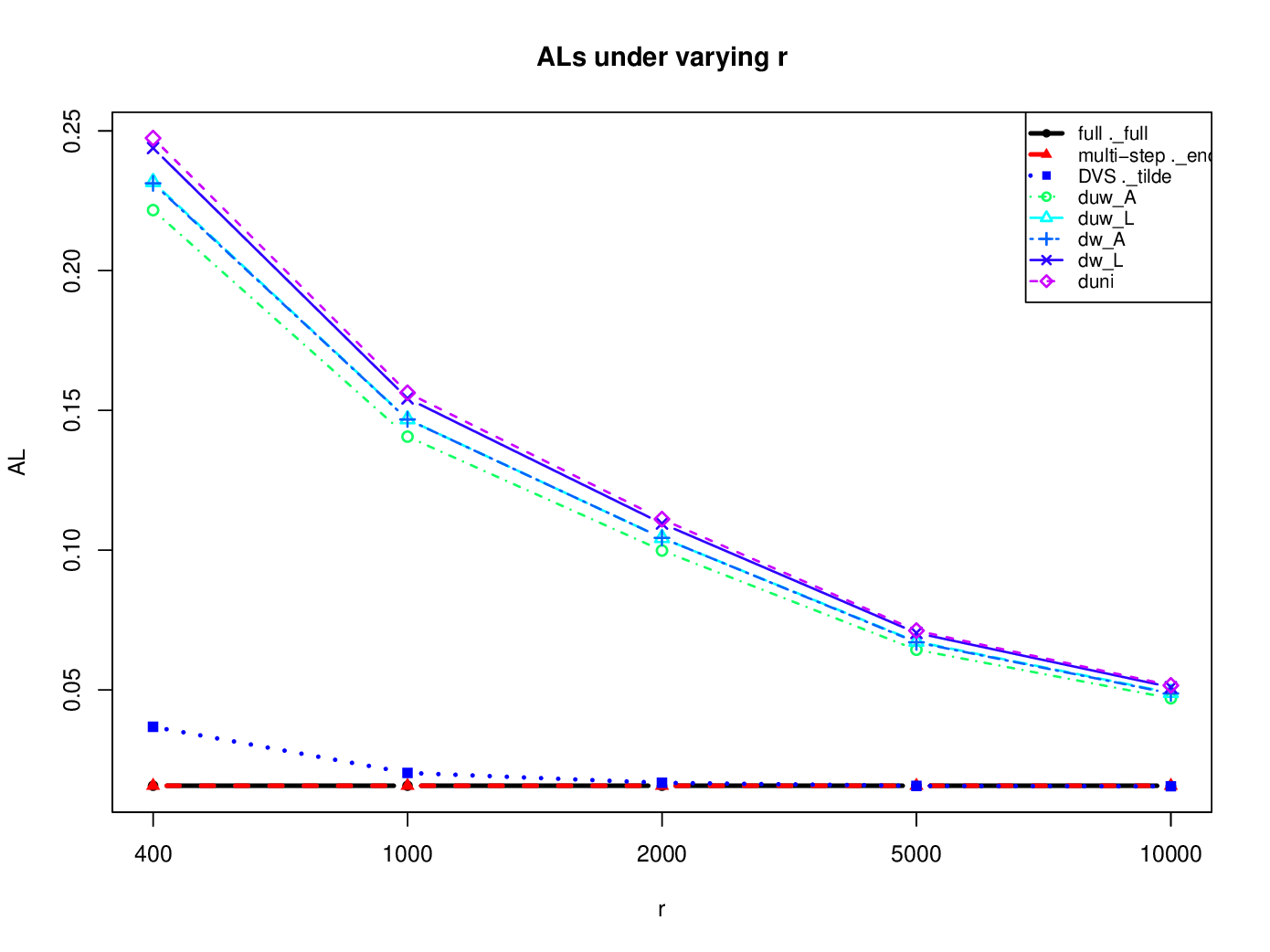}
	\end{subfigure}
	\caption{Comparison of ACPs and ALs under varying subsample sizes ($r$) for the linear model (Case a) with $n=10^5$, $p=500$, $d=5$, and $r_\mathrm{p}=2 \times 10^3$.}
	\label{fig1}
\end{figure}
Figure \ref{fig1} demonstrates that both the multi-step and DVS estimators consistently outperform the competing methods. The multi-step estimator attains average interval lengths nearly identical to the full-data method, while the DVS estimator exhibits a rapid decrease in interval length as the subsample size $r$ increases, achieving comparable performance when $r=2000$. 
These results further confirm the remarkable efficiency of our proposed approach. The experimental results for logistic regression are consistent with those for linear regression; further details are provided in the supplementary materials.

\subsection{Simultaneous inference}\label{sec5.2}
Let the sample size $n$, dimension $p$, and the pilot subsample size be $5\times 10^5$, $500$ and $5000$ in our method, respectively. The parameters of interest are the first $d$ variables with $d=50$. The confidence level is fixed as $\alpha=0.05$. The comparisons of the MSE, running time, ACP and AL of the $95\%$ simultaneous confidence intervals for four cases are shown in Table \ref{table2}.

As shown in Table \ref{table2}, the estimators proposed in Section \ref{sec4} provide a computationally efficient approach for simultaneous inference of the high-dimensional target parameters. Specifically, the empirical coverage probabilities of all methods remain close to $95\%$, and their average interval lengths and MSEs are similar, whereas our approach requires only about $10\%$ of the computational time of the baseline method, since merely $10\%$ of the data are utilized to obtain the initial estimator. Therefore, our estimators serve as a an effective and practical alternative to existing methods, as they offer comparable performance with substantially lower computational cost. The results for logistic regression also closely mirror those for linear regression in the supplementary materials.
\begin{table*}[htbp]
	\caption{MSE, Time (seconds), ACP, and AL for simultaneous inference in Linear model with $n=5\times10^4$, $p=500$, $r_\mathrm{p}=5000$, and $d=50$.}
	\resizebox{\linewidth}{!}{
		\centering
		\begin{tabular}{cccccccccc}
			\hline
			&$\check{\boldsymbol{\theta}}_{ns}$&$\check{\boldsymbol{\theta}}_{s}$&$\hat{\boldsymbol{\theta}}_{ns}$&$\hat{\boldsymbol{\theta}}_{s}$& &$\check{\boldsymbol{\theta}}_{ns}$&$\check{\boldsymbol{\theta}}_{s}$&$\hat{\boldsymbol{\theta}}_{ns}$&$\hat{\boldsymbol{\theta}}_{s}$\cr
			\hline
			&\multicolumn{4}{c}{(a)}& &\multicolumn{4}{c}{(b)}\cr
			MSE&$1.690\times10^{-3}$&$1.689\times10^{-3}$&$1.649\times10^{-3}$&$1.649\times10^{-3}$& &0.011&0.011&0.011&0.011\cr
			Time&128.89&130.69&1255.97&1261.22& &134.94&138.41&1431.81&1439.47\cr
			ACP&0.940&0.945&0.945&0.950& &0.932&0.936&0.932&0.932\cr
			AL&0.037&0.037&0.037&0.037& &0.097&0.097&0.096&0.096\cr
			&\multicolumn{4}{c}{(c)}& &\multicolumn{4}{c}{(d)}\cr
			MSE&$1.514\times10^{-3}$&$1.487\times10^{-3}$&$1.311\times10^{-3}$&$1.311\times10^{-3}$& &0.010&0.010&0.009&0.009\cr
			Time&131.73&142.49&1444.75&1452.32& &138.97&136.71&1438.51&1443.10\cr
			ACP&0.936&0.924&0.940&0.932& &0.960&0.956&0.952&0.952\cr
			AL&0.036&0.035&0.033&0.033& &0.092&0.092&0.086&0.086\cr
			\hline
		\end{tabular}
	}
	\label{table2}
\end{table*}

\section{Real data analysis}\label{sec6}
In this Section, we apply our proposed methods to
a superconductivity dataset \citep{hamidieh2018data} that is available from \href{https://archive.ics.uci.edu/ml/datasets/Superconductivty+Data#}{https://archive.ics.uci.edu/ml/datasets/Superconductivty+Data$\#$}. The dataset contains $21263$ observations and $81$ features, with the critical temperature of superconductors serving as the response variable.
\begin{table*}[!htbp]
	\caption{Bias, standard error (SD), Time (seconds), and confidence interval (CI) for $z_1$ and $z_2$ in the superconductivity dataset.}
	\resizebox{\linewidth}{!}{
		\centering
		\begin{tabular}{ccccccccccccccccc}
			\hline
			$r_{\mathrm{p}}$& &$\hat{\boldsymbol{\theta}}_{(end)}$&$\tilde{\boldsymbol{\theta}}$&$\check{\boldsymbol{\theta}}_{duw}^{A}$&$\check{\boldsymbol{\theta}}_{duw}^{L}$&$\check{\boldsymbol{\theta}}_{dw}^{A}$&$\check{\boldsymbol{\theta}}_{dw}^{L}$&$\check{\boldsymbol{\theta}}_{duni}$&\cr
			\hline
			\multirow{7}{*}{500}&Time&--&0.356&0.340&0.335&0.366&0.367&0.352\cr
			&Bias1&--&-0.076&-0.126&-0.130&-0.154&-0.085&-0.125\cr
			&SD1&--&\textbf{0.039}&0.050&0.066&0.072&0.062&0.061\cr
			&CI1&--&[-0.133,0.414]&[-2.040, 1.312]&[0.784,1.537]&[0.140,0.700]&[0.214,0.574]&[0.381,0.668]\cr
			&Bias2&--&-0.021&-0.020&0.052&-0.019&0.091&-0.022\cr
			&SD2&--&\textbf{0.030}&0.058&0.041&0.051&0.068&0.059\cr
			&CI2&--&[0.200,0.798]&[-1.776,3.767]&[-8.820,5.144]&[-0.023,0.524]&[-0.361,0.434]&[0.094,0.584]\cr
			\cr
			\multirow{7}{*}{1000}&Time&--&0.378&0.376&0.368&0.372&0.357&0.364\cr
			&Bias1&--&-0.006&-0.116&-0.023&-0.076&0.068&-0.030\cr
			&SD1&--&\textbf{0.037}&0.046&0.052&0.059&0.052&0.059\cr
			&CI1&--&[0.117,0.838]&[-0.568,0.448]&[-0.352,1.757]&[-0.433,0.821]&[0.363,1.098]&[0.223,1.296]\cr
			&Bias2&--&0.025&-0.019&0.048&0.019&-0.060&0.021\cr
			&SD2&--&\textbf{0.037}&0.061&0.053&0.040&0.052&0.042\cr
			&CI2&--&[0.250,0.790]&[-0.282,1.113]&[0.008,0.809]&[-0.310,0.256]&[-0.249,0.451]&[-0.146,0.226]\cr
			\cr
			\multirow{7}{*}{2000}&Time&--&0.410&0.408&0.424&0.424&0.436&0.418\cr
			&Bias1&--&-0.073&-0.070&-0.072&-0.036&-0.016&-0.056\cr
			&SD1&--&0.036&0.039&0.043&\textbf{0.031}&0.044&0.061\cr
			&CI1&--&[0.181,0.448]&[0.494,0.666]&[0.436,0.632]&[0.166,0.547]&[-0.156,0.748]&[-1.601,6.612]\cr
			&Bias2&--&0.017&0.082&0.068&-0.076&0.057&0.013\cr
			&SD2&--&0.023&0.033&0.037&\textbf{0.019}&0.038&0.045\cr
			&CI2&--&[0.129,0.355]&[0.060,0.315]&[0.104,0.335]&[0.059,0.406]&[0.057,0.348]&[-0.652,0.650]\cr
			\cr
			\multirow{7}{*}{5000}&Time&0.586&0.557&0.796&0.869&0.876&0.892&0.905\cr
			&Bias1&-0.003&-0.014&-0.074&-0.043&0.020&-0.002&0.022\cr
			&SD1&\textbf{0.008}&0.017&0.013&0.013&0.018&0.014&0.042\cr
			&CI1&[0.407,0.526]&[0.414,0.568]&[0.392,0.722]&[0.095,0.452]&[0.540,0.732]&[0.363,1.098]&[0.172,1.011]\cr
			&Bias2&-0.003&$4.773\times10^{-4}$&0.019&-0.008&0.030&0.053&0.007\cr
			&SD2&0.009&0.017&\textbf{0.008}&0.012&0.012&0.013&0.033\cr
			&CI2&[0.195,0.301]&[0.124,0.247]&[0.124,0.368]&[0.156,0.521]&[0.033,0.422]&[-34.62,33.00]&[-0.129,1.113]\cr
			\hline
		\end{tabular}
	}
	\label{table5}
\end{table*}

First, we focus on the weighted mean of thermal conductivity ($z_1$) and the range of atomic radius ($z_2$) since they are important features (see Table 5 of \citet{hamidieh2018data}). As a benchmark, we employing the decorrelated score approach \citep{ning2017general}, the estimates of $z_1$ and $z_2$ are $0.509$ and $0.241$, respectively, with a computational time of $1.514$ seconds. The corresponding $95\%$ confidence intervals are $[0.455,0.564]$ and $[0.194,0.288]$, respectively, suggesting that superconductors with higher weighted mean of thermal conductivity and range of atomic radius tend to exhibit higher critical temperatures. Now we compare the estimators in Section \ref{sec5.1} as follows. The multi-step estimator $\hat{\boldsymbol{\theta}}_{(end)}$ is applied with $r_\mathrm{p}=5000$, while $r=r_\mathrm{p}$ is chosen as $\{500,1000,2000,5000\}$ to leverage the DVS estimator $\tilde{\boldsymbol{\theta}}$, and other baseline works $\check{\boldsymbol{\theta}}_{dw}^{A}$, $\check{\boldsymbol{\theta}}_{dw}^{L}$, $\check{\boldsymbol{\theta}}_{duw}^{A}$, $\check{\boldsymbol{\theta}}_{duw}^{L}$, and $\check{\boldsymbol{\theta}}_{duni}$. Bias relative to the full-sample decorrelated estimates, standard deviation, and computational time (over $200$ replications) are reported in Table \ref{table5}, with $95\%$ confidence intervals from one replication.
\begin{table*}[!htbp]
	\caption{$95\%$ simultaneous confidence intervals for $20$ important features in the superconductivity dataset.}
	\resizebox{\linewidth}{!}{
		\centering
		\begin{tabular}{cccccccc}
			\hline
			Methods& & & & & & & \cr
			\hline
			\addlinespace[2pt]&1&2&3&4&5&6&7\\
			\multirow{6}{*}{$\check{\boldsymbol{\theta}}_{ns}$}&[-0.112, -0.109]
			&[0.074, 0.077]
			&[0.236, 0.239]
			&[-0.342, -0.339]
			&[0.262, 0.264]
			&[-68.769, -68.766]
			&[-0.574, -0.571]\cr
			\addlinespace[2pt]
			&8&9&10&11&12&13&14\\
			&[11.444, 11.446]
			&[-0.001, 0.001]
			&[-0.608, -0.605]
			&[-0.125, -0.122]
			&[0.547, 0.550]
			&[-4.742, -4.740]
			&[-0.352, -0.349]\cr
			\addlinespace[2pt]
			&15&16&17&18&19&20 \\
			&[0.003, 0.006]
			&[7.681, 7.684]
			&[-0.233, -0.230]
			&[-0.445, -0.442]
			&[-0.071, -0.068]
			&[-0.001, 0.002]\cr
			\addlinespace[2pt]&1&2&3&4&5&6&7\\
			\multirow{6}{*}{$\check{\boldsymbol{\theta}}_{s}$} 
			&[-0.089, -0.081]
			&[0.108, 0.116]
			&[0.142, 0.150]
			&[-0.177, -0.169]
			&[0.162, 0.170]
			&[-60.209, -60.201]
			&[-0.525, -0.518]\cr
			\addlinespace[2pt]
			&8&9&10&11&12&13&14\\
			&[20.530, 20.538]
			&[-0.004, 0.004]
			&[-0.611, -0.603]
			&[-0.125, -0.117]
			&[0.410, 0.418]
			&[-1.906, -1.898]
			&[-0.282, -0.274]\cr
			\addlinespace[2pt]
			&15&16&17&18&19&20\\
			&[0.002, 0.005]
			&[3.111, 3.118]
			&[-0.232, -0.225]
			&[-0.075, -0.067]
			&[-0.196, -0.188]
			&[-0.001, 0.001]\cr
			\addlinespace[2pt]&1&2&3&4&5&6&7\\
			\multirow{6}{*}{$\hat{\boldsymbol{\theta}}_{ns,\text{zhang}}$}
			&[-5.062, 4.421]
			&[-4.198, 5.285]
			&[-4.964, 4.519]
			&[-4.453, 5.030]
			&[-3.697, 5.786]
			&[19.132, 28.615]
			&[-4.380, 5.103]\cr
			\addlinespace[2pt]
			&8&9&10&11&12&13&14\\
			&[-68.544, -59.061]
			&[-58.346, -48.864]
			&[-4.277, 5.206]
			&[-5.251, 4.232]
			&[-4.830, 4.653]
			&[56.744, 66.227]
			&[-6.755, 2.728]\cr
			\addlinespace[2pt]
			&15&16&17&18&19&20\\
			&[-4.722, 4.761]
			&[34.088, 43.571]
			&[-5.395, 4.088]
			&[-3.677, 5.806]
			&[-3.579, 5.904]
			&[-4.753, 4.730]\cr
			\addlinespace[2pt]&1&2&3&4&5&6&7\\
			\multirow{6}{*}{$\hat{\boldsymbol{\theta}}_{s,\text{zhang}}$} 
			&[-0.335, -0.306]
			&[0.507, 0.580]
			&[-0.251, -0.195]
			&[0.248, 0.328]
			&[0.996, 1.094]
			&[17.981, 29.767]
			&[0.295, 0.427]\cr
			\addlinespace[2pt]
			&8&9&10&11&12&13&14\\
			&[-68.239, -59.366]
			&[-59.304, -47.906]
			&[0.361, 0.567]
			&[-0.531, -0.488]
			&[-0.157, -0.019]
			&[57.198, 65.772]
			&[-2.144, -1.882]\cr
			\addlinespace[2pt]
			&15&16&17&18&19&20\\
			&[0.018, 0.021]
			&[34.640, 43.018]
			&[-0.692, -0.616]
			&[0.970, 1.158]
			&[1.075, 1.249]
			&[-0.013, -0.011]\cr
			\hline
		\end{tabular}
	}
	\label{table6}
\end{table*}

Table \ref{table5} shows that the proposed DVS estimator generally achieves the smallest SD, particularly for subsample sizes $500$ or $1000$, indicating more stable estimates than \citet{shan2024optimal} under limited subsampling. Under the same $r_\mathrm{p}$, all methods incur similar computational costs and are much faster than the full-sample decorrelated estimators. However, the DVS estimator demonstrates superior inferential performance. Notably, when $r_\mathrm{p}=1000$, the DVS estimator is the only method correctly identifying both $z_1$ and $z_2$ as significantly positive at the $0.05$ level, consistent with the full-sample benchmark.

We further conduct simultaneous inference on the superconductivity dataset. As noted by Table 5 of \citet{hamidieh2018data}, $20$ features (ordered by importance) are important to the critical temperature of the superconductors. We first fit a Lasso model on the full dataset and obtain the estimates: $(-0.084,0.006,0.227,-0.312,0.253,-68.804,-0.568,
5.459,-0.284,-0.639,\\-0.157,0.513,-4.201,-0.357,
0.005,2.079,-0.216,-0.461,-0.134,0.001)$. Then we conduct simultaneous inference on the $20$ features. Denote the non-studentized and studentized versions of our estimators as $\check{\boldsymbol{\theta}}_{ns}$ and $\check{\boldsymbol{\theta}}_{s}$, respectively, and their counterparts in the baseline approach of \citet{zhang2017simultaneous} as $\hat{\boldsymbol{\theta}}_{ns,\text{zhang}}$ and $\hat{\boldsymbol{\theta}}_{s,\text{zhang}}$. The average running times across $100$ replications are $4.784$, $4.704$, $18.438$, and $18.464$ seconds for $\check{\boldsymbol{\theta}}_{ns}$, $\check{\boldsymbol{\theta}}_{s}$, $\hat{\boldsymbol{\theta}}_{ns,\text{zhang}}$, and $\hat{\boldsymbol{\theta}}_{s,\text{zhang}}$, respectively, indicating that our methods are substantially faster. The inference results are presented in Table \ref{table6}. Table \ref{table6} shows that the confidence intervals constructed by our proposed estimators closely align with the Lasso estimates, whereas those obtained by the methods of \citet{zhang2017simultaneous} deviate considerably. It also aligns with our previously inferential results, where $z_1$ corresponds to the $12$th and $z_2$ corresponds to the $3$rd feature. This discrepancy may stem from the strong multicollinearity among the covariates in the dataset-for instance, the dataset contains the geometric mean of density as an important feature and the mean of density in the remaining features that are highly correlated-which may undermine the inference accuracy of the method in \citet{zhang2017simultaneous}. In contrast, our approach alleviates this issue to some extent through the use of decorrelated score functions, which mitigates the influence of nuisance parameters. Since $\check{\boldsymbol{\theta}}_{s}$ yields feature-specific confidence intervals with adaptive lengths, we focus on its results. Except for the $9$th and $20$th features, all other $18$ features are identified as significant. Among them, the $1$st, $4$th, $6$th, $7$th, $10$th, $11$th, $13$th, $14$th, $17$th, $18$th, and $19$th features are significantly negative, while the $2$nd, $3$rd, $5$th, $8$th, $12$th, $15$th, $16$th features are significantly positive. Hence, superconductors with small values of the negative features and large values of the positive features are more likely to exhibit high critical temperatures.
\section{Conclusion}\label{sec7}

In this paper, we developed a unified and computationally efficient inference framework for both low- and high-dimensional parameters in large-scale GLMs. We propose two estimators—a DVS estimator and a multi-step estimator—that balance statistical efficiency and computational feasibility.  The DVS estimator enables fast inference with a small subsample, while both estimators achieve statistical performance comparable to full-data approaches  when $r/\sqrt{n} \to \infty$. For high-dimensional simultaneous inference, where conventional decorrelated score methods are inapplicable, we establish uniform Bahadur representations and asymptotic normality under general conditions. This enables valid simultaneous confidence intervals for large parameter sets and massive datasets, overcoming key statistical and computational challenges. Future work will focus on improving inference under measurement constraints and extending to prediction under covariate shift.


\section*{Disclosure statement}\label{disclosure-statement}
The authors declare that they have no conflicts of interest.
\section*{Acknowledgement}\label{Acknowledgement}
The authors are supported by Key technologies for coordination and interoperation of power distribution service resource under Grant No. 2021YFB2401300, and NSFC under Grant No. 12571311 and 12326606.

\section*{Supplementary material}
The supplementary material includes the additional simulation results for the Logistic regression model and the technical proofs. The code is made available at \href{https://github.com/BoFuxjtu/BoFuxjtu_DVS.git}{https://github.com/BoFuxjtu/BoFuxjtu$\_$DVS.git}.

  \bibliography{bibliography.bib}

\end{document}